\documentclass[11pt,a4paper]{article}
\usepackage[a4paper,hmargin=2.8cm,vmargin=3cm]{geometry}
\usepackage[utf8x]{inputenc}
\usepackage{amsmath,amsthm,amsfonts,amssymb}
\usepackage{authblk}

\usepackage[bookmarksnumbered=true]{hyperref}

\let\oldmarginpar\marginpar
\renewcommand\marginpar[1]{\-\oldmarginpar[\raggedleft\footnotesize #1]%
  {\raggedright\footnotesize #1}}

\newtheorem{thm}{Theorem}[section]  

\newtheorem{prp}[thm]{Proposition}
\newtheorem{lem}[thm]{Lemma}

\newcommand{\di}{{d}}		
\newcommand{\cc}[1]{\overline{#1}}	
\renewcommand{\Re}{\operatorname{Re}\,} 	

\newcommand{\norm}[1]{\lVert#1\rVert}	
\newcommand{\Tr}{\operatorname{tr}}	
\newcommand{\tr}{\operatorname{tr}}

\newcommand{\e}{\varepsilon}
\newcommand{\h}{\mathfrak{h}}
\newcommand{\eps}{\varepsilon}

\newcommand{\bR}{{\mathbb R}}

\newcommand{\bx}{{\bf{x}}}

\newcommand{\be}{\begin{equation}}
\newcommand{\ee}{\end{equation}}

\newcommand{\cJ}{{\cal J}}

\newcommand{\cF}{{\cal F}}
\newcommand{\cH}{{\cal H}}
\newcommand{\cL}{{\cal L}}
\newcommand{\cE}{{\cal E}}
\newcommand{\cN}{{\cal N}}
\newcommand{\cU}{{\cal U}}

\newcommand{\cQ}{{\cal Q}}

\newcommand{\re}{{\text{Re} }}

\newcommand{\xx}{{\mathbf x}}
\newcommand{\yy}{{\mathbf y}}
\newcommand{\ww}{{\mathbf w}}

\newcommand{\bN}{{\mathbb N}}

\newcommand{\ph}{{\varphi}}

\newcommand{\wt}{\widetilde}

\def\a{{\alpha}}

\def\g{{\gamma}}
\def\r{{\rho}}
\def\s{{\sigma}}
\def\o{{\omega}}
\def\n{{\nu}}
\def\D{{\Delta}}
\def\O{{\Omega}}
\def\G{{\Gamma}}

\def \bea{\begin{eqnarray}}
\def \eea{\end{eqnarray}}
\def\nn{{\nonumber}}

\title{Mean-field Evolution of Fermionic Mixed States}
\author{Niels Benedikter$^{*}$, Vojkan Jak\v{s}i\'c$^{**}$, Marcello Porta$^{\dag}$, Chiara Saffirio$^{\dag}$ 

and Benjamin Schlein$^{\dag}$ \\ $^{*}$Mathematics and Physics Department, Universit\`a Roma 3, L.go S. L. Murialdo 1, 00146 Roma, Italy \vspace{.2cm}\\ $^{**}$Mathematics and Statistics Department, McGill University, 805 Sherbrooke Street West, Montreal, QC, H3A 2K6, Canada \vspace{.2cm}\\ $^{\dag}$Mathematics Department, University of Z\"urich, Winterthurerstrasse 190, 8057 Z\"urich, Switzerland}

\begin{document}
\maketitle

\begin{abstract}
In this paper we study the dynamics of fermionic mixed states in the mean-field regime. We consider initial states which are close to quasi-free states and prove that, under suitable assumptions on the inital data and on the many-body interaction, the quantum evolution of such initial data is well approximated by a suitable quasi-free state. In particular we prove that the evolution of the reduced one-particle density matrix converges, as the number of particles goes to infinity, to the solution of the time-dependent Hartree-Fock equation. Our result holds for all times, and gives effective estimates on the rate of convergence of the many-body dynamics towards the Hartree-Fock evolution.
\end{abstract}


\section{Introduction}
\setcounter{equation}{0}

The time-evolution of many-body quantum systems in the so called mean-field regime has received a lot of attention in the recent  years. For bosonic systems, the mean-field Schr\"odinger dynamics can be generated by an $N$-particle Hamiltonian 
\begin{equation}\label{eq:HNbos} H_N = \sum_{j=1}^N -\Delta_{x_j} + \frac{1}{N} \sum_{i<j}^N V(x_i -x_j) 
\end{equation}
acting on the bosonic  Hilbert space $L^2_s (\bR^{3N}, dx_1 \dots dx_N)$ in the limit of large $N$. The coupling constant  $N^{-1}$ in front of the potential energy guarantees that both terms in the Hamiltonian are typically of the same size which is  a necessary condition  for the existence of  a non-trivial limiting evolution.

\medskip

The physical picture one has in mind is that  particles are initially trapped in a volume of order one and are  in  equilibrium. When the traps are switched off, the system starts to evolve, following the evolution generated by (\ref{eq:HNbos}). Thus,  the object of interest is the dynamics $\psi_{N,t} = e^{-i H_N t} \psi_N$ where the  initial data $\psi_N$ describes trapped particles at equilibrium. At low temperature, the system exhibits complete condensation and  it is natural to consider  approximately factorized initial data $\psi_N \simeq \ph^{\otimes N}$ for an appropriate $\ph \in L^2 (\bR^3)$. It turns out 
that in the limit of large $N$ the evolution preserves the  approximate factorization. The relation $\psi_{N,t} \simeq \ph_t^{\otimes N}$ then easily leads 
to the  self-consistent nonlinear Hartree equation 
\[ i\partial_t \ph_t = - \Delta \ph_t + (V*|\ph_t|^2) \ph_t \]
for the evolution of the condensate wave function. 

\medskip

The convergence of the many-body Schr\"odinger evolution towards the Hartree dynamics has to be understood in the sense of reduced density matrices. The one-particle reduced density $\gamma^{(1)}_{N,t}$ associated with the solution $\psi_{N,t}$ of the Schr\"odinger equation is given by the partial trace
\[ \gamma^{(1)}_{N,t} = N \tr_{2,\dots , N} |\psi_{N,t} \rangle \langle \psi_{N,t}|, \]
where $|\psi_{N,t} \rangle \langle \psi_{N,t}|$ indicates the orthogonal projection onto $\psi_{N,t}$. In other words, $\gamma^{(1)}_{N,t}$ is the non-negative trace class operator on $L^2 (\bR^3)$ with the  integral kernel
\[ \gamma^{(1)}_{N,t} (x,y) = N \int dx_2 \dots dx_N \, \psi_{N,t} (x, x_2, \dots , x_N) \overline{\psi}_{N,t} (y,x_2, \dots , x_N) .\]
Analogously, we can define the $k$-particle reduced density $\gamma^{(k)}_{N,t}$ associated with $\psi_{N,t}$ by
\begin{equation}\label{eq:gammak} \gamma^{(k)}_{N,t} = \binom{N}{k} \, \tr_{k+1, \dots , N} |\psi_{N,t} \rangle \langle \psi_{N,t}| .
\end{equation}
Notice that we choose the normalization 
\[ \tr \gamma^{(k)}_{N,t} = \binom{N}{k} .\]
Of course, if $k < N$, the $k$-particle reduced density $\gamma^{(k)}_{N,t}$ does not contain the full information about the state $\psi_{N,t}$; however, knowledge of $\gamma^{(k)}_{N,t}$ is sufficient to compute the expectation of any $k$-particle observable. In fact, if $J^{(k)}$ is an operator acting on $L^2 (\bR^{3k})$, and if we denote by $J^{(k)}_{i_1, \dots , i_k}$ the operator acting as $J^{(k)}$ on the $k$ particles $i_1, \dots , i_k$ and as the identity on the other $(N-k)$ particles, we have
\[ \sum_{(i_1, \dots , i_k)} \langle \psi_{N,t} , J^{(k)}_{i_1, \dots, i_k} \psi_{N,t} \rangle = \tr \, J^{(k)} \gamma_{N,t}^{(k)}, \]
where, on the left hand side, we are summing over all clusters of $k$ particles.

\medskip

Reduced density matrices provide the appropriate language to express the convergence of the many-body Schr\"odinger dynamics towards the limiting Hartree equation. Notice that the $k$-particle reduced density matrix associated with the product state $\ph^{\otimes N}$ is given by 
\[ \binom{N}{k} |\ph \rangle \langle \ph|^{\otimes k}  = \binom{N}{k} |\ph^{\otimes k} \rangle \langle \ph^{\otimes k}|. \]

\medskip

Under appropriate conditions on the interaction potential $V$, and assuming that, at time $t=0$, 
\[ \lim_{N\to \infty} \frac{1}{N} \gamma^{(1)}_{N,0} = |\ph \rangle \langle \ph| \]
one can show that, for every $t \in \bR$, 
\begin{equation}
\label{eq:bos-conv} \lim_{N\to \infty}\frac{1}{N} \gamma^{(1)}_{N,t} = |\ph_t \rangle \langle \ph_t| .
\end{equation}
Similarly, one also gets convergence of the reduced densities $\gamma^{(k)}_{N,t}$, for any fixed $k \in \bN$. Eq.\ (\ref{eq:bos-conv}) states the stability of condensation, understood as approximate factorization of the wave function, with respect to the time evolution. 

\medskip

The first results in the direction of (\ref{eq:bos-conv}) were obtained in \cite{Hepp,GV,Spohn}. More recently, the case of a Coulomb interaction has been studied in \cite{EY,BGM}. Different points of view 
on the bosonic mean-field limit have been proposed in \cite{FKS,FKP} (where the convergence towards the Hartree equation has been interpreted as an Egorov-type theorem) and in \cite{AN} (focusing on the mean-field propagation of Wigner measures). Bounds on the rate of convergence towards the Hartree equation have been proven in \cite{RS,KP,CLS} and fluctuations around the Hartree dynamics have been analyzed in \cite{GMM,BKS,BSS,LNS}. Convergence towards the defocusing nonlinear Schr\"odinger equation with local cubic nonlinearity has been established in \cite{AGT} for one-dimensional systems, in \cite{KSS} for the two-dimensional case and in \cite{ESY1} in three dimensions. In one dimension, the focusing NLS has been recently derived in \cite{CH1}. The more singular Gross-Pitaevskii regime has been considered in \cite{ESY2,ESY3,P,BDS,CH2}. In the case of three-body interactions, the convergence towards the nonlinear Schr\"odinger equation with a quintic nonlinearity has been established in \cite{CP,XC}. 

\medskip

For fermionic systems, the mean-field limit is more delicate. Consider a gas of $N$ fermions described by the Hamiltonian 
\[ H_N = \sum_{j=1}^N -\Delta_{x_j} + \lambda\sum_{i<j}^N V(x_i -x_j) \]
acting on the fermionic  Hilbert space $L^2_a (\bR^{3N}, dx_1 \dots dx_N)$ (the subspace of $L^2 (\bR^{3N})$ consisting of permutation antisymmetric functions satisfying 
\[ \psi_N (x_{\pi(1)}, \dots , x_{\pi(N)}) = \sigma_\pi \psi_N (x_1, \dots , x_N) \]
for every permutation $\pi \in S_N$, where $\sigma_\pi \in \{\pm 1 \}$ denotes the sign of the permutation $\pi$). Because of the Pauli principle (reflected in the antisymmetry of the wave functions), if the $N$ particles are trapped in a volume of order one, their kinetic energy is (at least) of the order $N^{5/3}$. To make sure that the potential energy is of the same order, one has to take  $\lambda = N^{-1/3}$, which is much larger than $\lambda=N^{-1}$ 
in the bosonic case. Moreover, since the typical velocity of the particles is very large, of the order $N^{1/3}$, one can only hope to follow the time evolution for times of the order $N^{-1/3}$. After rescaling time, we are led to the $N$-particle Schr\"odinger equation
\begin{equation}\label{eq:schr0} iN^{1/3} \partial_t \psi_{N,t} = \left[ \sum_{j=1}^N -\Delta_{x_j} + \frac{1}{N^{1/3}} \sum_{i<j}^N V(x_i -x_j) \right] \psi_{N,t} .
\end{equation}
In these variables, we are interested in the solution $\psi_{N,t}$ for times $t$ of order one. Setting $\eps = N^{-1/3}$ and multiplying (\ref{eq:schr0}) with $\eps^2$, we conclude that 
\begin{equation}\label{eq:ferm-schr}
i\eps \partial_t \psi_{N,t} = \left[ \sum_{j=1}^N -\eps^2 \Delta_{x_j} + \frac{1}{N} \sum_{i<j}^N V(x_i -x_j) \right] \psi_{N,t}.
\end{equation}
Hence, we recover the $N^{-1}$ coupling constant characterizing the mean-field limit in the bosonic setting. Here, however, the mean-field regime is naturally linked with a semiclassical limit, with $\eps = N^{-1/3}$ playing the role of Planck's constant. Let us remark, however, that 
different scaling limits in the fermionic setting have been considered in \cite{BGGM,FK,PP}.

\medskip

Similarly to the bosonic case, we are interested in the solution of (\ref{eq:ferm-schr}) for initial data $\psi_{N,0}$ describing fermions confined to  a volume of order one and in thermal  equilibrium. 

\medskip

At zero temperature, the ground state of the system of trapped fermions is expected (and in certain cases, it is  known) to be approximated by a Slater determinant
\[ \psi_\text{slater} (x_1, \dots , x_N) = \det \left( f_i (x_j) \right) \]
for an appropriate set of $N$ orthonormal orbitals $f_1, \dots, f_N \in L^2 (\bR^3)$. The Slater determinant is completely characterized by its one-particle reduced density 
\[ \omega_N = N \tr_{2,\dots, N} |\psi_\text{slater} \rangle \langle \psi_\text{slater}| = \sum_{j=1}^N |f_j \rangle \langle f_j|, \]
given by the orthogonal projection onto the $N$-dimensional 
space spanned by $f_1, \dots , f_N$. In the presence of an external potential $V_\text{ext}$, the energy of the Slater determinant with reduced density $\omega_N$ is given by the Hartree-Fock functional
\begin{equation}\label{eq:HF-fun} \cE_\text{HF} (\omega_N) = \tr \, \left(-\eps^2 \Delta + V_\text{ext} (x) \right) \omega_N + \frac{1}{2N} \int dx dy V(x-y) \left[ \omega_N (x,x) \omega_N (y,y) - |\omega_N (x,y)|^2 \right] .
\end{equation}
Semiclassical analysis suggests that, in the limit of large $N$, the minimizer of (\ref{eq:HF-fun}) can be approximated by the Weyl quantization 
\begin{equation}\label{eq:OpW} \text{Op}^W_M (x,y) = \frac{1}{(2\pi)^3} \int dp \, M \left( \frac{x+y}{2} , p \right) e^{i p \cdot \frac{x-y}{\eps}} 
\end{equation}
of the phase-space density $M(x,p) = \chi (|p| \leq c \rho^{1/3}_\text{TF} (x))$, where $\rho_{TF}$ is the minimizer of the Thomas-Fermi functional 
\begin{equation}\label{eq:TF-fun} \cE_\text{TF} (\rho) = \frac{3}{5} \e^{2} c_{TF}\int \rho^{5/3} (x) dx + \int V_\text{ext} (x) \rho (x) dx + \frac{1}{2N} \int dx dy\, V(x-y) \rho (x) \rho (y)  
\end{equation}
under the condition $\int \rho (x) dx = N$. Eq.\ (\ref{eq:OpW}) implies in particular that the integral kernel of the Hartree-Fock minimizer $\omega_N$ should vary on scales of order one in the $x+y$ direction and on much shorter scales, of order $\eps$, in the $x-y$ direction. Accordingly, we expect that
\begin{equation}\label{eq:semi}
\begin{split} 
\tr \, \left| [ x, \omega_N ] \right| &\leq C N \eps, \\ \tr \, \left| [\eps \nabla, \omega_N ] \right| &\leq C N \eps.   \end{split} 
\end{equation}
In fact, it is easy to show that (\ref{eq:semi}) holds true for the Weyl quantization (\ref{eq:OpW}) under reasonable assumptions on the Thomas-Fermi density $\rho_{\text{TF}}$. 

\medskip

It is therefore interesting to study the solution $\psi_{N,t}$ of the many-body Schr\"odinger equation (\ref{eq:ferm-schr}) for initial data close to a Slater determinant with one-particle reduced density satisfying (\ref{eq:semi}). This was done in \cite{EESY}, for analytic interaction and for short times\footnote{In \cite{EESY}, the estimates (\ref{eq:semi}) are not explicitly stated, but are hidden in the assumption on the convergence of the Wigner transform of the initial data.}, and, more recently, in \cite{BPS} for arbitrary times and for a much larger class of regular interactions (a similar analysis for semi-relativistic fermions can be found in \cite{BPS2}). In \cite{BPS}, the one-particle reduced density $\gamma^{(1)}_{N,t}$ associated with $\psi_{N,t}$ was proven to satisfy 
\begin{equation}\label{eq:conv-fer} \| \gamma^{(1)}_{N,t} - \omega_{N,t} \|_\text{HS} \leq C \exp (c \exp (c|t|)) 
\end{equation}
uniformly in $N$. Here $\omega_{N,t}$ is the solution of the time-dependent Hartree-Fock equation
\begin{equation}\label{eq:HF} i\eps \partial_t \omega_{N,t} = \left[ -\eps^2 \Delta + V*\rho_t + X_t , \omega_{N,t} \right] 
\end{equation}
with initial data $\omega_{N,0} = \omega_N$. Here we used the notation $\rho_t (x) = N^{-1} \omega_{N,t} (x,x)$ for the normalized density associated with $\omega_{N,t}$ and $X_t$ to denote the exchange operator defined by the integral kernel 
\[ X_t (x,y) = N^{-1} V(x-y) \omega_{N,t} (x,y) \] 
The right hand side of (\ref{eq:conv-fer}) should be compared with the Hilbert-Schmidt norm of $\gamma_{N,t}^{(1)}$ and of $\omega_{N,t}$, which is of the order $N^{1/2}$ (it is also possible to get bounds for the trace norm of the difference $\gamma^{(1)}_{N,t} - \omega_{N,t}$, see \cite{BPS}). In other words, (\ref{eq:conv-fer}) states that the evolution of an initial Slater determinant satisfying the semiclassical bounds (\ref{eq:semi}) remains close to a Slater determinant, whose reduced one-particle density evolves according to (\ref{eq:HF}). Notice that the assumption (\ref{eq:semi}) for the initial data plays a crucial role in the analysis. Physically, it is motivated by the fact that we are interested in initial data describing systems at equilibrium. 

\medskip

In contrast with the bosonic case, the limiting Hartree-Fock equation (\ref{eq:HF}) still depends on $N$. As $N \to \infty$, the Hartree-Fock dynamics converges, at least formally, towards the Vlasov dynamics. More precisely, let 
\[ W_{N,t} (x,v) = \frac{\eps^{3}}{(2\pi)^3} \int dy \, \gamma_{N,t}^{(1)} \left(x+ \frac{\eps y}{2}, x- \frac{\eps y}{2} \right) e^{i y \cdot v} \]
be the Wigner transform of the one-particle reduced density $\gamma^{(1)}_{N,t}$ associated with the solution of the Schr\"odinger equation (\ref{eq:ferm-schr}). Then a simple computation suggests that, as $N \to \infty$, $W_{N,t} \to W_{\infty,t}$, where $W_{\infty,t}$ satisfies the Vlasov equation
\begin{equation}\label{eq:vlasov} \partial_t W_{\infty,t} + v \cdot \nabla_x W_{\infty,t} + \nabla \left(V* \rho_t \right) \cdot \nabla_v W_{\infty,t} = 0 
\end{equation}
with the density $\rho_t (x) = \int dv\, W_{\infty,t} (x,v)$. 
Convergence of the many-body evolution towards the Vlasov equation (\ref{eq:vlasov}) was first established, for analytic interaction potentials, in \cite{NS}. Later, this result was extended to a larger class of potentials in \cite{Sp}. Also (\ref{eq:conv-fer}) implies convergence (in a weak sense) towards the Vlasov dynamics. It should be noted, however, that establishing the convergence (\ref{eq:conv-fer}) is more complicated than comparing the solution of the Schr\"odinger equation directly with the Vlasov equation, and requires a deeper understanding of the many-body dynamics.

\medskip

To establish the convergence (\ref{eq:conv-fer}) and also to generalize the analysis to positive temperatures, it is useful to switch to a Fock space representation of the fermionic system. The fermionic Fock space over $L^2 (\bR^3)$ is defined as the direct sum
\[ \cF = \bigoplus_{n \geq 0} L^2_a (\bR^{3n}, dx_1 \dots dx_n). \]
In particular, states with exactly $N$ particles are described on the Fock space by sequences $\{ 0, 0, \dots, \psi_N, 0, \dots \}$ with only one non-vanishing component. On $\cF$, Slater determinants can be conveniently obtained by applying Bogoliubov transformations to the Fock vacuum $\Omega = \{ 1, 0, \dots \}$. In fact, for every orthogonal projection $\omega_N$ on $L^2 (\bR^3)$ with $\tr \omega_N = N$ there exists a unitary operator $R_{\omega_N} : \cF \to \cF$ implementing a Bogoliubov transformation such that $R_{\omega_N} \Omega$ is exactly the $N$-particle Slater determinant with reduced one-particle density $\omega_N$.

\medskip

In the present paper, we are interested in the dynamics at positive temperature. More precisely, we are interested in initial data describing interacting fermions trapped in a volume of order one, in equilibrium at a certain temperature $T > 0$ and with an appropriate chemical potential $\mu$, so that the average number of particles is $N$, the same parameter appearing in the Hamiltonian generating the many-body evolution (the precise definition of the Hamiltonian will be given in the next section; notice however, that the Hamiltonian is the same as the one considered at zero temperature). 

\medskip

In our setting it is expected that equilibrium states at positive temperature can be approximated by quasi-free states on $\cF$. Quasi-free states are completely characterized by their one-particle reduced density; all higher order correlations can be expressed in terms of the one-particle reduced density using Wick's theorem (in general, one also needs the pairing density, but here we will assume it to vanish). In fact, every operator $\omega_N$ with $0 \leq \omega_N \leq 1$ and $\tr \omega_N = N$ is the one-particle reduced density of a unique quasi-free state on $\cF$. A simple example of quasi-free states are Slater determinants; in this case, $\omega_N$ is additionally an orthogonal projection. However, there is an important difference between Slater determinants (which are important at zero temperature) and the quasi-free states which are relevant at positive temperature. Slater determinants are {\it pure} quasi-free states. At positive temperature, on the other hand, we will have to consider {\it mixed} quasi-free states on $\cF$. This difference makes the analysis at positive temperature more involved, compared with \cite{BPS}. 

\medskip

Mixed states are described by density matrices (and not simply by vectors) on the fermionic Fock space $\cF$. For our purposes, it will be useful to represent mixed states on $\cF$ through vectors on the double Fock space $\cF \otimes \cF \simeq \cF (L^2 (\bR^3) \oplus L^2 (\bR^3))$. In this way, we will be able to obtain any mixed quasi-free state by applying a Bogoliubov transformation, this time acting on the vacuum of the double Fock space $\cF (L^2 (\bR^3) \oplus L^2 (\bR^3))$ (this representation of mixed quasi-free states is a special case of the  Araki-Wyss construction, see \cite{AW,DG}).  

\medskip

Similarly to (\ref{eq:semi}), we will need the (approximate) initial quasi-free state to satisfy certain estimates, which are justified by the fact that we 
are starting from an equilibrium state and that  we observe its evolution resulting from a change of the external field ({\it e.g.} removal of the trapping potential). Semiclassical analysis suggests that equilibrium states at positive temperature $T > 0$ can be approximated by the Weyl quantization 
\begin{equation}\label{eq:We-pos}
\text{Op}^W_M (x,y) = \frac{1}{(2\pi)^3} \int dp \, M \left( \frac{x+y}{2}, p \right) e^{ip \cdot \frac{x-y}{\eps}} 
\end{equation} 
of a phase-space density 
\begin{equation} \label{eq:Mpos} M (x,p) = g_{T,\mu} \left(p^2 - c \rho_\text{TF}^{2/3} (x) \right) 
\end{equation}
for the smooth Fermi-Dirac distribution \[ g_{T,\mu} (E) = \frac{1}{1+ e^{(E-\mu)/T}} \] depending on the temperature $T>0$ and on the chemical potential $\mu\geq 0$. In (\ref{eq:Mpos}), $\rho_{\text{TF}}$ denotes the minimizer of the Thomas-Fermi functional (\ref{eq:TF-fun}) and the normalization constant $c>0$ has to be chosen so that 
\[ \int M(x,v) dx dv = 1 \]
Similarly to (\ref{eq:OpW}), the kernel (\ref{eq:We-pos}) decays to zero, for $|x-y| \gtrsim \eps$. Since the function $g_{T,\mu}$ for $T > 0$ is smoother than at zero temperature ($g_{T,\mu}$ is a characteristic function for $T=0$), we observe that (\ref{eq:We-pos}) satisfies even stronger semiclassical estimates, compared with (\ref{eq:semi}). We will be interested in approximately quasi-free states, with one-particle reduced density matrix $0 \leq \omega_N \leq 1$ with $\tr \omega_N = N$ and satisfying the bounds
\begin{equation}\label{eq:semi+}
\begin{split} \left\| \left[ \sqrt{\omega_N} , e^{i p \cdot x} \right] \right\|_{\text{HS}}&\leq C \eps N^{1/2} (1+|p|),\\[2mm]
 \left\| \left[ \sqrt{1-\omega_N} , e^{i p \cdot x} \right] \right\|_{\text{HS}} &\leq C \eps N^{1/2} (1+|p|), \\[2mm]
 \left\| \left[ \sqrt{\omega_N}, \eps \nabla \right] \right\|_\text{HS} &\leq C \eps N^{1/2},\\[2mm] 
 \left\| \left[ \sqrt{1-\omega_N}, \eps \nabla \right] \right\|_\text{HS} &\leq C \eps N^{1/2}.
\end{split} 
\end{equation}
It is then easy to check that the Weyl quantization (\ref{eq:We-pos}) satisfies the estimates (\ref{eq:semi+}), under reasonable assumptions on the Thomas-Fermi density $\rho_{\text{TF}}$. Our main result is that for  such initial data the time evolution remains close to a quasi-free state, with reduced one-particle density $\omega_{N,t}$ satisfying the time-dependent Hartree-Fock equation (\ref{eq:HF}). 

\medskip

{\it Acknowledgements.} M. Porta, C. Saffirio and B. Schlein have been supported by the ERC grant MAQD-240518 and by the NCCR SwissMAP. N. Benedikter has been partially supported by the ERC grant CoMBoS-239694. The research of V. Jak\v{s}i\'c was partly supported by NSERC. We are grateful to the Erwin Schr\"odinger Institute in Vienna; part of this work has been completed during the workshop on ``Scaling limits in classical and quantum mechanics''.

\section{Fock space representation, quasi-free states and main result}
\setcounter{equation}{0}

{\it Fock space.} Let $\h = L^2 (\bR^3)$. We define the fermionic Fock space $\cF (\h)$ over $\h$ as the direct sum 
\be
\cF(\h) = \bigoplus_{n\geq 0}\h^{\wedge n}
\ee
where $\h^{\wedge n} := \h \wedge \h \cdots \wedge \h$ is the $n$-fold antisymmetric tensor product of $\h$. Vectors in $\cF (\h)$ are sequences 
\be
\cF(\h) \ni\xi = (\xi^{(0)},\,\xi^{(1)},\, \ldots,\, \xi^{(n)},\,\ldots)
\ee
with $\xi^{(n)}\in \h^{\wedge n}$. In particular, the vacuum $\O_{\cF (\h)} = (1, 0, \dots )$ describes a state with no particles. $\cF (\h)$ is a Hilbert space, with respect to the inner product 
\[ \left\langle \xi , \chi \right\rangle = \sum_{n \geq 0} \langle \xi^{(n)}, \chi^{(n)} \rangle_{\h^{\wedge n}}   \]

The creation and annihilation operators on $\cF(\h)$ are defined as follows:  for $f \in \h$,
\[ \begin{split} (a^* (f) \psi)^{(n)} (x_1, \dots , x_n) &= \frac{1}{\sqrt{n}} \sum_{j=1}^n (-1)^j  f (x_j) \psi^{(n-1)} (x_1, \dots , x_{j-1}, x_{j+1}, \dots , x_n) , \\ 
(a(f) \psi)^{(n)} (x_1 , \dots , x_n) &= \sqrt{n+1} \int dx \overline{f (x)} \psi^{(n+1)}(x, x_{1},\ldots, x_{n})
\end{split}\]
As the notation suggests, the creation operator $a^* (f)$ is the adjoint of the annihilation operator $a(f)$. Observe that $a^* (f)$ is linear, while $a(f)$ is antilinear in the argument $f \in \h$.  Creation and annihilation operators satisfy canonical anticommutation relations 
\begin{equation}\label{eq:CAR} \{ a (f) , a^* (g) \} = \langle f , g \rangle_{\h}\;, \quad \{ a (f) , a(g) \} = \{ a^* (f) , a^* (g) \} = 0 \end{equation}
for all $f,g \in \h$. Here $\{ A , B \} = AB + BA$ denotes the anticommutator of the operators $A,B$. Notice that (in contrast to the bosonic operators) fermionic creation and annihilation operators are {\it bounded}, with 
\begin{equation}\label{eq:bdaa*}  \| a^\sharp (f) \| = \| f \|  
\end{equation}
where $a^\sharp$ is either $a$ or $a^*$. 

It is also useful to introduce the operator-valued distributions $a^{*}_{x}$ and $a_{x}$, which formally create and annihilate, respectively, a particle at point $x \in \bR^3$. They 
satisfy
\[
a(f) = \int dx\, a_{x}\overline{f(x)},\qquad a^{*}(f) = \int dx\, a^{*}_{x}f(x)
\]
for all $f\in\h$. 

\medskip

Given a self-adjoint operator $O: \h \to \h$, we define its second quantization $d\G (O): \cF \to \cF$ by  
\be\label{eq:0}
(d\G(O)\psi)^{(n)} = \sum_{j\leq n} O^{(j)}\psi^{(n)}\;,
\ee
where $O^{(j)}= 1\otimes \dots 1 \otimes O \otimes 1 \otimes \dots \otimes 1$ is the operator acting as $O$ on the $j$-th particle and as the identity on the other $(n-1)$ particles. An important example is the particle number operator $\cN = d\G(1)$, which counts the number of particles in each sector of the Fock space: $(\cN \psi)^{(n)} = n\psi^{(n)}$. If the one-particle operator $O$ has the integral kernel $O(x;y)$, then $d\Gamma (O)$ can be expressed in terms of operator-valued distributions as 
\[
d\G(O) = \int dxdy\, O(x;y) a^{*}_{x}a_{y}.
\]
In particular, the number operator can be written as
\[
\cN = \int dx\, a^{*}_{x}a_{x}.
\]

{\it Mixed states.} A fermionic mixed state is represented by a density matrix on $\cF (\h)$. A density matrix is a non-negative trace class operator $\r:\cF (\h) \to \cF(\h)$ with $\tr \, \r = 1$. Notice that the state described by the density matrix $\rho$ is pure if and only if $\rho$ is a rank-$1$ orthogonal projection onto a $\psi \in \cF$, {\it i.e.} $\rho = |\psi \rangle \langle \psi|$. In general, 
\be\label{eq:densma}
 \r = \sum_n \lambda_n |\psi_{n} \rangle \langle \psi_{n}| 
 \ee
where $\lambda_n \geq 0$, $\{\psi_n\}$ is an orthonormal set in ${\cal F}(\h)$, and  $\sum_{n} \lambda_n = 1$. Physically, $\r$ describes an incoherent superposition of pure states and  $\lambda_n$ is the probability that the system is  in the state $\psi_n$. Obviously, the expectation of an arbitrary operator $A$ on $\cF (\h)$ with respect to the mixed state with density matrix $\r$ is 
\[ \tr \, A \r = \sum_n \lambda_n \langle \psi_n, A \psi_n \rangle. \] 

Given the density matrix (\ref{eq:densma})
we define the operator $\wt{\kappa}:\cF(\h) \to \cF(\h)$ by 
\[
\wt{\kappa} = \sum_{n} \e_{n}| \psi_{n} \rangle \langle \phi_{n} |
\]
where  $\e_{n}\in \mathbb{C}$ is a sequence satisfying  $|\e_{n}|^{2} = \lambda_{n}$ and $\{\phi_{n}\}$ is an orthonormal set in ${\cal F}(\h)$. Clearly,
\[
\wt{\kappa} \wt{\kappa}^{*} = \r
\]
Of course, such a decomposition of $\rho$ is far from being unique and later we shall choose a convenient one. Since $\r$ is trace class, $\wt{\kappa} \in \cL^2 (\cF (\h))$, the set of Hilbert-Schmidt operators on $\cF (\h)$. 

Next, we observe that $\cL^2 (\cF (\h)) \simeq \cF (\h) \otimes \cF (\h)$. The isomorphism is induced by the map $|\psi \rangle \langle \phi| \to \psi \otimes \overline{\phi}$, extended by linearity to the whole space $\cL^2 (\cF (\h))$. The mixed state with density matrix (\ref{eq:densma}) is described on $\cF (\h) \otimes \cF (\h)$ by the vector
\[ \kappa = \sum_{n} \e_n \psi_n \otimes \overline{\phi}_n. \]
The expectation of the operator $A$ on $\cF (\h)$ in the state described by $\kappa \in \cF (\h) \otimes \cF (\h)$ is then given by
\begin{equation}\label{eq:trAr} \tr \, A \r = \tr \, A \wt{\kappa} \wt{\kappa}^* = \langle \kappa, (A \otimes 1) \kappa \rangle_{\cF (\h) \otimes \cF (\h)} \, . 
\end{equation}

It is well-known that the doubled Fock space $\cF(\h)\otimes \cF(\h)$ is isomorphic to the Fock space $\mathcal{F}(\h \oplus \h)$ (see \cite{DG} or 
any book on mathematical quantum field theory). The unitary map $U$ that implements this isomorphism is known as the {\it exponential law} and it is defined by the relations 
\[
U (\O_{\cF(\h)}\otimes\O_{\cF(\h)}) = \O_{\mathcal{F}(\h \oplus \h)}
\]
and
\be\label{eq:a00}
\begin{split}
U \left[ a(f)\otimes 1 \right] U^{*} &= a(f\oplus 0) =: a_{l}(f) \\ 
U \left[ (-1)^{\cN} \otimes a(f) \right] U^{*} &= a(0\oplus f) =: a_{r}(f) \end{split} \ee
for all $f\in \h$, where $a_{\s}(f)$, $\s = l,\,r$, are the {\it left and right representations} of $a(f)$, respectively. By hermitian  conjugation, we also find 
\begin{equation}\label{eq:a00*} \begin{split}
U \left[ a^{*}(f)\otimes 1 \right] U^{*} &= a^{*}(f\oplus 0) =: a^{*}_{l}(f) \\ U \left[ (-1)^{\cN}\otimes a^{*}(f) \right] U^{*} &= a^{*}(0\oplus f) =: a^{*}_{r}(f)
\end{split}
\end{equation}
where $a^{*}_{\s}(f)$, $\s = l,\,r$, are the left and right representations of $a^{*}(f)$. Notice that the presence of the operator $(-1)^{\cN}$ on the second line of (\ref{eq:a00}) and (\ref{eq:a00*}) guarantees that creation operators on the space $\cF (\h \oplus \h)$ satisfy the canonical anticommutation relations (and in particular that $a^\sharp_l (f)$ anticommutes with $a^\sharp_r (g)$, for all $f,g \in \h$). 

It is convenient to introduce the left and right representations of the operator-valued distributions $a_{x}$, $a^{*}_{x}$ by the relations
\be\label{eq:a01}
\begin{split}
a_{l}(f) &= \int dx\, a_{x,l}\overline{f(x)},\qquad a_{r}(f) = \int dx\, a_{x,r}\overline{f(x)}, \\
a^{*}_{l}(f) &= \int dx\, a^{*}_{x,l} f(x),\qquad a^{*}_{r}(f) = \int dx\, a^{*}_{x,r} f(x),
\end{split}
\ee
for all $f\in \h$. 

We also define the left and right representations of the second quantization of operators on $\h$ by 
\[ \begin{split} 
U \left[ d\G(O)\otimes 1 \right] U^{*} &= d\G(O\oplus 0) =: d\G_{l}(O) 
\\ U \left[ 1\otimes d\G(O) \right] U^{*} &= d\G(0\oplus O) =: d\G_{r}(O)\;.\end{split} \]
The left and right representations of $d\G(O)$ can be written in terms of the operator-valued distributions as
\be
d\G_{l}(O) = \int dxdy\, O(x;y)a^{*}_{x,l}a_{y,l},\qquad d\G_{r}(O) = \int dxdy\, O(x;y)a^{*}_{x,r}a_{y,r}.
\ee

{\it Reduced densities.} According to (\ref{eq:trAr}), the expectation of the observable $d\G (O)$ in the mixed state described by $\psi \in \cF (\h \oplus \h)$ is given by
\[ \langle \psi, d\G_l (O) \psi \rangle = \int dx dy\, O(x,y) \langle \psi, a_{x,l}^* a_{y,l} \psi \rangle = \tr \, O \, \gamma^{(1)}_\psi ,\]
where $\gamma_\psi^{(1)}$ is the one-particle reduced density associated with $\psi$, defined as the non-negative trace-class operator on $\h$ with the integral kernel 
\begin{equation}\label{eq:gam1-FF} \gamma^{(1)}_\psi (x;y) = \langle \psi, a_{y,l}^* a_{x,l} \psi \rangle .
\end{equation}

It is also useful to introduce the pairing density $\alpha_\psi$ associated with $\psi$. It is defined as the Hilbert-Schmidt operator on $\h$ with integral kernel 
\[ \alpha_\psi (x;y) = \langle \psi, a_{y,l} a_{x,l} \psi \rangle .\]

It is convenient to combine the reduced one-particle density $\gamma^{(1)}_\psi$ and the pairing density $\alpha_\psi$ in a single non-negative operator $\Gamma_\psi : \h \oplus \h \to \h \oplus \h$ acting as 
\begin{equation}\label{eq:Gdef0} \Gamma_\psi = \left( \begin{array}{ll} \gamma_\psi^{(1)} & \alpha_\psi \\ \alpha_\psi^{*} & 1- \overline{\gamma}_{\psi}^{(1)} \end{array} \right) .
\end{equation}
$\Gamma_\psi$ is known as the generalized one-particle reduced density associated with the mixed state described by $\psi \in \cF (\h \oplus \h)$. It allows us to compute the expectation of arbitrary quadratic expressions in creation and annihilation operators on $\cF (\h)$. In fact, for any $f_1, f_2, g_1 , g_2 \in \h$, we find 
\begin{equation}\label{eq:Gdef} \left\langle \psi, \left( a_l^* (f_1) + a_l (\overline{g}_1) \right) \left(a_l (f_2) + a_l^* (\overline{g}_2) \right)  \psi \right\rangle = \langle (f_1, g_1), \Gamma_\psi (f_2, g_2) \rangle 
\end{equation}
where, on the r.h.s., $\langle . , . \rangle$ denotes the inner product on $\h \oplus \h$.  Using (\ref{eq:Gdef}), it is easy to check that $0 \leq \Gamma_\psi \leq 1$ for any $\psi \in \cF (\h \oplus \h)$. 

To compute the expectation of $k$-particle observables (given by the product of more than two creation and annihilation operators), we need to define higher order correlation functions. We define the $k$-particle reduced density $\gamma^{(k)}_\psi$ associated with $\psi \in \cF (\h \oplus \h)$ as the non-negative trace-class operator on $\h^{\wedge k}$ with the integral kernel 
\begin{equation}\label{eq:gammak-F} \gamma^{(k)}_\psi (x_1, \dots , x_k ; y_1, \dots , y_k) = \frac{1}{k!} \left\langle \psi, a_{y_1, l}^* \dots a_{y_k,l}^* a_{x_k, l} \dots a_{x_1, l} \psi \right\rangle .
\end{equation}
Notice that we use the normalization 
\[ \tr \gamma^{(k)}_{\psi} = \frac{1}{k!} \langle \psi, \cN_l (\cN_l-1) \dots (\cN_l-k+1) \psi \rangle\]
with $\cN_l = d\G_l (1)$. Hence, for states with average number of particles $N$, we expect $\tr \gamma^{(k)}_\psi$ to be of the order $N^k$ (for states with exactly $N$ particles we have $\tr \gamma^{(k)}_\psi = \binom{N}{k}$, consistently with the definition (\ref{eq:gammak})). 

\medskip

{\it Time evolution of mixed states.} On $\cF$, we define the Hamilton operator $\cH_N$ by
\[ (\cH_N \Psi)^{(n)} = \cH_N^{(n)} \Psi^{(n)} \]
with 
\[ \cH_N^{(n)} = \sum_{j=1}^n -\eps^2 \Delta_{x_j} + \frac{1}{N} \sum_{i<j}^n V(x_i- x_j) .\]
By definition, $\cH_N$ commutes with the particle number operator, and, when restricted to the $N$-particle sector,  coincides with the Hamiltonian generating the evolution (\ref{eq:ferm-schr}). 
It is sometimes useful to express the Hamiltonian $\cH_N$ in terms of the operator-valued distributions $a_x,a_x^*$. We find
\begin{equation}\label{eq:ham-F}
\cH_N = \eps^2 \int dx \nabla_x a_x^* \nabla_x a_x + \frac{1}{2N} \int dx dy V(x-y ) a_x^* a_y^* a_y a_x .
\end{equation}

The time evolution of the density matrix $\r$ is given by \[ \r_{t} = e^{-i\mathcal{H}_{N}t/\e}\r e^{i\mathcal{H}_{N}t/\e}.\]  
Thus, we define the time evolution of $\kappa \in \cF (\h) \otimes \cF (\h)$ by $\kappa_{t} = e^{-i\mathcal{H}_{N}t/\e} \otimes e^{i\cH_N t/\eps} \kappa$. We denote by $\psi_t = U \kappa_t$ the vector in $\cF (\h \oplus \h)$ describing the mixed state with density matrix $\rho_t$. Then 
\[
\begin{split}
\psi_{t} &= U \kappa_{t} = Ue^{-i(\cH_{N}\otimes 1 - 1\otimes \cH_{N})t/\e}\kappa = e^{-i \cL_{N}t/\e} \psi
\end{split}
\]
where the {\it Liouvillian} $\cL_{N}$ is defined by 
\[
\cL_{N} = U\big(\mathcal{H}_{N}\otimes 1 - 1\otimes \mathcal{H}_{N}\big)U^{*}.
\]
A more  explicit expression follows from (\ref{eq:a00}) and  (\ref{eq:a01}):
\be\label{a0c}
\begin{split}
\cL_{N} =\;& \e^{2}\int dx\, \nabla_{x}a^{*}_{x,l}\nabla_{x}a_{x,l} + \frac{1}{2N}\int dxdy\, V(x-y)a^{*}_{x,l}a^{*}_{y,l}a_{y,l}a_{x,l}  \\
& - \e^{2}\int dx\, \nabla_{x}a^{*}_{x,r}\nabla_{x}a_{x,r} - \frac{1}{2N}\int dxdy\, V(x-y)a^{*}_{x,r}a^{*}_{y,r}a_{y,r}a_{x,r}.
\end{split}
\ee
Hence, the expectation of an arbitrary operator $A$ on $\cF (\h)$ in the evolved mixed state is given by
\[ \tr A \r_t = \langle \psi_t , (A\otimes 1) \psi_t \rangle = \langle \psi, e^{i\cL_N t/\e} (A \otimes 1) e^{-i\cL_N t/\e} \psi \rangle. \]

\medskip

{\it Bogoliubov transformations.} We will be interested in the time evolution of quasi-free initial data, which are supposed to be a good approximation of equilibrium states of confined systems in the mean field limit. It turns out that mixed quasi-free states can be represented  in $\cF (\h \oplus \h)$ by applying Bogoliubov transformations on the vacuum of $\cF (\h \oplus \h)$. 

Let us briefly review the concept of Bogoliubov transformation. For $f,g \in \h \oplus \h$ we define the field operators 
\[
A(f,g) = a(f) + a^{*}(\overline g), \qquad A^{*}(f,g) = (A(f,g))^{*} = a^{*}(f) + a(\overline g).
\]
We observe that  
\be\label{a3}
A^{*}(f,g) = A(Jg, Jf) = A (\cJ (f,g)) 
\ee
where $J : \h \oplus \h \to \h \oplus \h$ is the antilinear operator defined by $J f = \overline{f}$ and $\cJ : (\h \oplus \h) \oplus (\h \oplus \h) \to (\h \oplus \h) \oplus (\h \oplus \h)$ is given by 
\[ \cJ = \left( \begin{array}{ll} 0 & J \\ J & 0 \end{array} \right). \]
The anticommutation relations (\ref{eq:CAR}) can be rewritten as 
\be\label{a4}
\{ A(f_{1},g_{1}),\, A^{*}(f_{2},g_{2}) \} = \langle (f_{1},\,g_{1}),\, (f_{2},g_{2}) \rangle
\ee
for all $f_1, f_2, g_1, g_2 \in \h \oplus \h$. Here $\langle \cdot, \cdot \rangle$ denotes the inner product on $(\h\oplus \h) \oplus (\h\oplus \h)$. 

\medskip

A Bogoliubov transformation is a linear map $\nu : (\h \oplus \h) \oplus (\h \oplus \h) \to (\h \oplus \h) \oplus (\h \oplus \h)$ such that $\nu^* \nu = 1$ and 
\[ \nu \cJ = \cJ \nu .\]
These two conditions guarantee that the rotated field operators $B(f,g) := A (\nu (f,g))$ satisfy the same relations (\ref{a3}) and (\ref{a4}) as the original operators $A(f,g)$, for all $f,g \in \h \oplus \h$. 

It is easy to check that a map $\nu$ is a Bogoliubov transformation if and only if it is of the form
\be\label{eq:bogUV}
\n = \begin{pmatrix} U & \overline{V} \\ V & \overline{U} \end{pmatrix}
\ee
where $U,\,V: \h \oplus \h \to \h \oplus \h$ are linear maps with 
\be\label{a5b}
U^{*}U + V^{*}V = 1,\qquad U^{*}\overline{V} + V^{*}\overline U = 0
\ee

A Bogoliubov transformation $\nu : (\h \oplus \h) \oplus  (\h \oplus \h)$ is said to be unitary implementable if there exists a unitary operator $R_{\nu}:\cF (\h \oplus \h) \mapsto \cF (\h \oplus \h)$ such that 
\be\label{a6}
A(\nu(f,g)) = R^{*}_{\nu}A(f,g)R_{\nu}
\ee
for all $f,\,g \in \h \oplus \h$. The {\it Shale-Stinespring condition} states that the Bogoliubov transformation (\ref{eq:bogUV}) is implementable if and only if $V$ is a Hilbert-Schmidt operator. The proof of this theorem and more information about Bogoliubov transformations can be found in the lecture notes \cite{Solovej}. 

\medskip

{\it Quasi-free states.}  A vector $\psi \in \cF (\h \oplus \h)$ is said to describe a quasi-free state if for all $f_1, \dots , f_{2n+1} \in \h$, 
\[\langle \psi, a_l^{\sharp_1} (f_1) \dots a_l^{\sharp_{2n+1}} (f_{2n+1}) \psi \rangle=0,\]
and for all $f_1, \dots , f_{2n} \in \h$, the Wick's rule  
\begin{equation}\label{eq:wick} \langle \psi, a_l^{\sharp_1} (f_1) \dots a_l^{\sharp_{2n}} (f_{2n}) \psi \rangle = \sum_{\pi} \sigma_\pi  \prod_{j=1}^n \langle \psi, a_l^{\sharp_{i_j}} (f_{i_j}) a_l^{\sharp_{\ell_j}} (f_{\ell_j}) \psi \rangle 
\end{equation}
holds. The sum on the r.h.s. runs over all pairings $\pi$, mapping the $n$ indices $i_1, \dots , i_n$ into $\ell_1, \dots , \ell_n$. 

Hence, as long as we are interested in expectations of the physically relevant left - observables (i.e. observables expressed through left creation and annihilation operators), quasi-free states are completely characterized by their generalized one-particle reduced density (\ref{eq:Gdef0}). Notice that for every $\Gamma : \h \oplus \h \to \h \oplus \h$ satisfying $0 \leq \Gamma \leq 1$, there exists a quasi-free state on $\cF (\h \oplus \h)$ with generalized one-particle density matrix $\Gamma$; see \cite{Solovej}. 

It turns out that quasi-free states can be conveniently written in terms of Bogoliubov transformations.  
More precisely, the vector $\psi \in \cF (\h \oplus \h)$ describing a quasi-free state   can always be written as $\psi = R_{\nu}\O$, 
for a suitable Bogoliubov transformation $\nu$ on $(\h \oplus \h) \oplus (\h \oplus \h)$; see \cite{Solovej}. This is a particular example of a famous construction in quantum statistical mechanics called  Araki-Wyss representation \cite{AW} (see also \cite{DG} for additional information and references).

In the following, we will be interested in quasi-free states $\psi \in \cF (\h \oplus \h)$ with average number of particles equal to $N$ and with vanishing pairing density $\alpha_\psi = 0$. The justification for this restriction comes from physics; equilibrium states for systems of weakly interacting fermions (described in the grand canonical ensemble by the Hamiltonian (\ref{eq:ham-F})) trapped in a volume of order one can be approximated by quasi-free states with vanishing pairing density. In fact, the pairing density is only expected to play an important role when the potential varies on very short scales, of the order $\eps$ (in this limit, the Bardeen-Cooper-Schrieffer (BCS) theory is believed to be relevant and $\alpha$ can be significantly different from zero). The generalized one-particle density matrix of such quasi-free states has the form
\begin{equation}\label{eq:G1} \Gamma^{(1)} = \left( \begin{array}{ll} \omega_N & 0 \\ 0 & 1-\overline{\omega}_N \end{array} \right) 
\end{equation}
for a non-negative trace class operator $\omega_N$ on $\h$ satisfying $0 \leq \omega_N \leq 1$ and $\tr \omega_N = N$. 


It is then easy to check that the quasi-free state with generalized reduced density (\ref{eq:G1}) can be written as $\psi = R_{\nu_N} \Omega$, with the Bogoliubov transformation $\nu_N : (\h \oplus \h)\oplus (\h \oplus \h) \to (\h \oplus \h) \oplus (\h \oplus \h)$ defined by
\begin{equation}\label{eq:def-nuN} \nu_N = \left( \begin{array}{ll} U_N & \overline{V}_N \\ V_N & \overline{U}_N \end{array} \right), 
\end{equation}
where
\be\label{a8}
U_N = \begin{pmatrix} u_N & 0 \\ 0 & \overline{u}_N \end{pmatrix},\qquad V = \begin{pmatrix} 0 & \overline{v}_N \\ -v_N & 0 \end{pmatrix},
\ee
and  $u_N = \sqrt{1-\omega_N}$, $v_N = \sqrt{\omega_N}$ (notice that $U_N,V_N$ satisfy the relations (\ref{a5b})). Since $\tr\, V_N^{*}V_N = 2 \tr \omega_N = 2N<\infty$, the Bogoliubov transformation $\nu_N$ is implementable. 

\medskip

To check that the $\nu_N$ defined by (\ref{a8}) is indeed the correct Bogoliubov transformation, we  note first that, for any $f \in \h$, 
\be
\begin{split}
R^{*}_{\nu_N} a_{l}(f) R_{\nu_N} & = a_{l}(u_N f) - a^{*}_{r}(\overline{v}_N \overline{f}), \\ 
R^{*}_{\nu_N} a_{r}(f) R_{\nu_N} &= a_{r}(\overline{u}_N f) + a^{*}_{l}(v_N \overline f).
\end{split}
\ee
Equivalently,
\be\label{a8b}
\begin{split}
R^{*}_{\nu_N} a_{x,l} R_{\nu_N} & = a_{l}(u_{x}) - a^{*}_{r}(\overline{v}_{x}), \\ 
R^{*}_{\nu_N} a_{x,r} R_{\nu_N} &= a_{r}(\overline{u}_{x}) + a^{*}_{l}(v_{x}),
\end{split}
\ee
where we used the notation $u_x (y) = u_N (y,x)$, $v_x (y)= v_N (y,x)$, and where $u_N (y,x)$ and $v_N (y,x)$ denote the kernels of the operators $u_N, v_N$. 

Using (\ref{a8b}) and the unitarity of $R_{\nu_N}$ we find 
\be
\begin{split}
\g_{\psi}(x;y) &= \langle R_{\nu_N}\O, a^{*}_{y,l}a_{x,l} R_{\nu_N}\O\rangle \\
& = \langle \O, R^{*}_{\nu_N}a^{*}_{y,l}R_{\nu_N}R^{*}_{\nu_N}a_{x,l}R_{\nu_N}\O\rangle \\
& = \langle \O, (a^{*}_{l}(u_{y}) - a_{r}(\overline{v}_{y}))(a_{l}(u_{x}) - a^{*}_{r}(\overline{v}_{x})) \O\rangle \\
& = \langle \O, a_{r}(\overline{v}_{y})a^{*}_{r}(\overline{v}_{x}) \O\rangle\\
& = (v_N^{*}v_N)(x;y)\\ & = \o_{N}(x;y)
\end{split}
\ee
where we used the  anticommutation relations and the fact that $a_{l},\, a_{r}$ annihilate the vacuum. Also,
\be
\begin{split}
\alpha_{\psi}(x;y) &= \langle R_{\nu_N}\O, a_{y,l}a_{x,l} R_{\nu_N}\O\rangle\\
&= \langle \O, (a_{l}(u_{y}) - a^{*}_{r}(\overline{v}_{y})(a_{l}(u_{x}) - a^{*}_{r}(\overline{v}_{x}))\O\rangle\\
&= \langle \O, a_{l}(u_{y})a^{*}_{r}(\overline{v}_{x})\O\rangle\\
&= 0
\end{split}
\ee
where we used that $\{a_{l}(u_{y}),\, a_{r}^{*}(\overline{v_{x}})\} = 0$. Finally, notice that the vector $\psi = R_{\nu_N}\O\in \cF(\h \oplus \h)$ corresponds to a specific choice of $\wt{\kappa}\in \cL^{2}(\cF(L^{2}))$ such that $\wt{\kappa}\wt{\kappa}^{*} = \rho$, where  $\rho$ is the quasi-free state with reduced one-particle density matrix given by $\o_{N}$. In fact, $\wt{\kappa} \simeq U^{*} R_{\nu_N} \O\in \cF(\h)\otimes \cF(\h)$; the isomorphism is the one induced by the map $\psi \otimes \overline{\phi} \to |\psi \rangle \langle \phi|$, extended by linearity to the whole space $\cF(\h)\otimes \cF(\h)$.


\medskip

{\it Main result.} In our main theorem, we consider the time evolution of initial mixed states which are close to quasi-free states
satisfying certain semiclassical estimates. For such initial data, we show that the evolution remains close to a (mixed) quasi-free state, with reduced density evolving according to the Hartree-Fock equation (\ref{eq:HF}). 
\begin{thm}\label{thm:main}
Let $V\in L^{1}(\bR^{3})$, and assume that
\begin{equation}\label{eq:Vass}
\int dp\, (1+|p|)^{2}|\widehat{V}(p)| < \infty.
\end{equation}
Let $\o_{N}$ be a sequence of operators on $\h = L^{2}(\bR^{3})$ such that $0\leq \o_{N}\leq 1$, $\tr(1-\D)\o_{N}<\infty$, $\tr\o_{N} = N$, and
\begin{equation}\label{eq:oass}
\begin{split}
\|[v_N, x]\|_\text{HS} &\leq CN^{1/2}\e, \qquad \|[v_N,\e\nabla]\|_\text{HS} \leq CN^{1/2}\e, \\
\|[u_{N}, x]\|_\text{HS} &\leq CN^{1/2}\e, \qquad  \|[u_N,\e\nabla]\|_\text{HS} \leq CN^{1/2}\e,
\end{split}
\end{equation}
with $v_{N} = \sqrt{\o_{N}}$, $u_{N} = \sqrt{1 - \o_{N}}$, for a suitable constant $C>0$. Let $\nu_{N}$ denote the Bogoliubov transformations (\ref{eq:def-nuN}), such that $R_{\nu_{N}}\O_{\cF (\h\oplus\h )}$ is the quasi-free state on $\cF (\h \oplus \h)$ with the generalized one-particle density matrix
\be
\G_{\n_{N}} = \begin{pmatrix} \o_{N} & 0 \\ 0 & 1 - \overline{\o}_{N} \end{pmatrix}.
\ee
Let $\xi_N \in \cF (\h \oplus \h)$ be a sequence with $\langle \xi_N, \cN^{10} \xi_N \rangle \leq C$ uniformly in $N$, and such that $\xi_{N} = \chi(\cN \leq CN)\xi_{N}$, where $\chi(\cdot)$ is the characteristic function and $C\geq 0$ is independent of $N$. Here $\cN = d\Gamma_l (1) + d\Gamma_r (1)$ is the particle number  operator in the extended Fock-space $\cF (\h \oplus \h)$. 

Let $\g^{(1)}_{N,t}$ be the one-particle reduced density matrix associated with the evolved state
\be
\psi_{N,t} = e^{-i\cL_{N}t/\e}R_{\nu_{N}}\xi_{N}
\ee
where the Liouvillian $\cL_{N}$ has been defined in (\ref{a0c}). Let $\o_{N,t}$ be the solution of the time-dependent Hartree-Fock equation
\be\label{eq:HF2}
i\e\partial_{t}\o_{N,t} = [-\e^{2}\D + V*\r_{t} - X_{t},\o_{N,t}]\;,
\ee
with the initial data $\o_{N,0} = \o_{N}$. Then, there exist constants $K,\,c_{1},\, c_{2}$, such that
\be\label{eq:conv-HS}
\big\| \g^{(1)}_{N,t} - \o_{N,t} \big\|_{HS} \leq K \exp(c_{2}\exp(c_{1}|t|))
\ee
and
\be\label{eq:conv-tr}
\tr \big| \g^{(1)}_{N,t} - \o_{N,t} \big| \leq K N^{1/2} \exp(c_{2}\exp(c_{1}|t|)).
\ee
\end{thm}

We can also prove convergence of higher order densities. The $k$-particle reduced density associated with the quasi-free state having one-particle reduced density $\omega_{N,t}$ is given by the $k$-fold wedge product 
\[ \omega^{\wedge k}_{N,t} (x_1, \dots , x_k ; y_1, \dots, y_k) = \sum_{\pi \in S_k} \sigma_\pi \prod_{j=1}^k \omega_{N,t} (x_j ; y_{\pi (j)}). \]

\begin{thm}\label{thm:maink}
Let $V \in L^1 (\bR^3)$ satisfy (\ref{eq:Vass}). Let $\omega_N$ be a sequence of trace-class operators on $\h$ with $0 \leq \omega_N \leq 1$, $\tr (1-\Delta)\o_{N}<\infty$, $\tr \omega_N = N$ and satisfying (\ref{eq:oass}). Fix $k \in \bN$ and let $\xi_N$ be a sequence in $\cF (\h \oplus \h)$ such that $\langle \xi_N, \cN^{10k} \xi_N \rangle \leq C$ uniformly in $N$, and such that $\xi_{N} = \chi(\cN \leq CN)\xi_{N}$ for $C\geq 0$ independent of $N$. Again, $\cN = d\Gamma_l (1) + d\Gamma_r (1)$. 

We denote by $\gamma^{(k)}_{N,t}$ the reduced $k$-particle density matrix (defined as in (\ref{eq:gammak-F})) associated with $\psi_{N,t} = e^{-i\cL_{N}t/\e}R_{\nu_{N}}\xi_{N}$.  Then there are constants $K,c_1, c_2$ depending also on $k \in \bN$, such that 
\be\label{eq:conv-HSk}
\big\| \g^{(k)}_{N,t} - \o^{\wedge k}_{N,t} \big\|_{HS} \leq K N^{(k-1)/2} \exp(c_{2}\exp(c_{1}|t|))
\ee
and
\be
\tr \big| \g^{(k)}_{N,t} - \o^{\wedge k}_{N,t} \big| \leq K N^{k - 1/2} \exp(c_{2}\exp(c_{1}|t|)).
\ee
\end{thm}

The general strategy of the proof is similar to the one in  \cite{BPS} for pure quasi-free states. The main new ingredient is the representation of fermionic mixed states on $\cF(\h)$ as pure states on $\cF(
\h\oplus \h)$. As explained above, on $\cF (\h \oplus \h)$ time-dependent quasi-free states can be represented using Bogoliubov transformations.  With $\o_{N,t}$ denoting the solution of the Hartree-Fock equation (\ref{eq:HF}), we define the corresponding Bogoliubov transformation $\n_{N,t} : (\h \oplus \h) \oplus (\h \oplus \h) \to (\h \oplus \h) \oplus (\h \oplus \h)$ 
by
\[
\nu_{N,t} = \begin{pmatrix} U_{N,t} & \overline{V}_{N,t} \\ V_{N,t} & \overline{U}_{N,t}  \end{pmatrix} 
\]
where the linear maps $U_{N,t},\, V_{N,t} : \h \oplus \h \to \h \oplus \h$ have the form 
\be
U_{N,t} = \begin{pmatrix} u_{N,t} & 0 \\ 0 & \overline{u}_{N,t} \end{pmatrix}, \qquad V_{N,t} = \begin{pmatrix} 0 & \overline{v}_{N,t} \\ -v_{N,t} & 0 \end{pmatrix}
\ee
with $u_{N,t} = \sqrt{1 - \o_{N,t}}$ and $v_{N,t} = \sqrt{\o_{N,t}}$. Since $\tr V^{*}_{N,t}V_{N,t} = 2N <\infty$, this Bogoliubov transformation is implementable, and we shall denote by $R_{t}\equiv R_{\nu_{N,t}}$ its unitary implementor on $\cF (\h\oplus \h)$.

\medskip

In our proof we compare $\psi_{N,t} = e^{-i \cL_N t/\eps} R_0 \xi_N$ with the quasi-free state $R_t \Omega$. To this end, we define the fluctuation vectors $\xi_{N,t}$ by
\[ \psi_{N,t} = e^{-i\cL_{N}t/\e}R_{0}\xi_{N} = R_t \xi_{N,t}. \]
Equivalently 
$\xi_{N,t} = \cU_N (t;0) \xi_N$, with the fluctuation dynamics 
\be\label{eq:flu-dyn}
\cU_{N}(t;s) = R^{*}_{t} e^{-i\cL_{N}(t - s)/\e} R_{s}
\ee
Similarly to \cite{BPS}, the problem of proving the  convergence (\ref{eq:conv-HS}) of the one-particle reduced density reduces to controlling the growth of the expectation of the particle number  operator $\cN = d\Gamma_l (1) + d\Gamma_r (1)$ with respect to the fluctuation dynamics. Analogously, the convergence (\ref{eq:conv-HSk}) of the $k$-particle reduced density follows from an estimate for the $k$-th moment of the particle number operator with respect to $\cU_N$. 

\medskip

The paper is organized as follows. In Section \ref{sec:gen} we compute the generator of $\cU_{N}$. In Section \ref{sec:growth} we control the growth of fluctuations with respect to the dynamics $\cU_{N}$. To do so we  first control the growth of fluctuations with respect to a suitable auxiliary dynamics $\widetilde{\cU}_{N}$, and then show  that $\widetilde{\cU}_{N}$ is close in norm to $\cU_{N}$. Then, in Section \ref{sec:proof} we use the bounds on the growth of fluctuations to prove Theorems \ref{thm:main} and \ref{thm:maink}. Finally, in Appendix \ref{secpropagation} we prove the propagation of the semiclassical structure of the initial data.

\section{The generator of the fluctuation dynamics}\label{sec:gen}
\setcounter{equation}{0}

The fluctuation dynamics $\cU_{N}(t;s)$ defined in (\ref{eq:flu-dyn}) satisfies
\be
i\e \frac{d}{dt} \cU_{N}(t;s) = \mathcal{G}_{N}(t) \cU_{N}(t;s)\;,\label{g1}
\ee
with $\cU_N (s;s) = 1$ for all $s \in \bR$ and where
\be
\mathcal{G}_{N}(t) = (i\e\partial_{t}R^{*}_{t})R_{t} + R^{*}_{t}\cL_{N}R_{t}.\label{g2}
\ee
In the next proposition, we compute explicitly the generator $\mathcal{G}_{N}(t)$.
\begin{prp}
The generator (\ref{g2}) of the fluctuation dynamics is given by
\be
\mathcal{G}_{N}(t) = d\Gamma_{l}(h_{HF}(t)) - d\Gamma_{r}(\overline{h_{HF}(t)}) + \mathcal{C_{N}} + \mathcal{Q}_{N}
\ee
where 
\begin{equation}\label{eq:CN}
\begin{split}
\mathcal{C}_{N} = \; & \frac{1}{2N}\int dxdy\, V(x-y)\Big( a^{*}_{l}(u_{x})a^{*}_{l}(u_{y})a_{l}(u_{y})a_{l}(u_{x}) + 2a^{*}_{l}(u_{x})a^{*}_{r}(v_{x})a_{r}(v_{y})a_{l}(u_{y}) \\
& -2a^{*}_{l}(u_{x})a^{*}_{r}(\overline{v_{y}})a_{r}(\overline{v_{y}})a_{l}(u_{x}) + a^{*}_{r}(\overline{v_{y}})a^{*}_{r}(\overline{v_{x}})a_{r}(\overline{v_{x}})a_{r}(\overline{v_{y}}) \\
& - a^{*}_{r}(\overline{u_{x}})a^{*}_{r}(\overline{u_{y}})a_{r}(\overline{u_{y}})a_{r}(\overline{u_{x}}) - 2a^{*}_{r}(\overline{u_{x}})a^{*}_{l}(v_{x})a_{l}(v_{y})a_{r}(\overline{u_{y}}) \\ & + 2a^{*}_{r}(\overline{u_{x}})a^{*}_{l}(v_{y})a_{l}(v_{y})a_{r}(\overline{u_{x}}) - a^{*}_{l}(v_{y})a^{*}_{l}(v_{x})a_{l}(v_{x})a_{l}(v_{y})\Big)
\end{split} 
\end{equation}
contains the quartic terms that commute with $\cN = d\G_{l}(1) + d\G_{r}(1)$, and 
\begin{equation}\label{eq:QN}
\begin{split} 
\cQ_{N} = \; & \frac{1}{2N}\int dxdy\, V(x-y)\Big(a^{*}_{l}(u_{x})a^{*}_{l}(u_{y})a^{*}_{r}(\overline{v_{y}})a^{*}_{r}(\overline{v_{x}}) + 
\\
&+ 2a^{*}_{l}(u_{x})a^{*}_{l}(u_{y})a^{*}_{r}(\overline{v_{x}})a_{l}(u_{y}) - 2a^{*}_{l}(u_{x})a^{*}_{r}(\overline{v_{y}})a^{*}_{r}(\overline{v_{x}})a_{r}(\overline{v_{y}}) - \\
& - a^{*}_{r}(\overline{u_{x}})a^{*}_{r}(\overline{u_{y}})a^{*}_{l}(v_{y})a^{*}_{l}(v_{x}) + 2a^{*}_{r}(\overline{u_{x}})a^{*}_{r}(\overline{u_{y}})a^{*}_{l}(v_{x})a_{r}(\overline{u_{y}}) - \\
& - 2a_{r}^{*}(\overline{u_{x}})a^{*}_{l}(v_{y})a^{*}_{l}(v_{x})a_{l}(v_{y}) + h.c.\Big)
\end{split}
\end{equation}
are the quartic terms that do not commute with $\cN$. In (\ref{eq:CN}) and (\ref{eq:QN}), we introduced the notation $u_x (y) := u_{N,t} (y,x)$, $v_x (y) := v_{N,t} (y,x)$ and similarly for the complex conjugate kernels. 
\end{prp}

\begin{proof} {\it Computation of $(i\e\partial_{t}R^{*}_{t})R_{t}$}. To evaluate the first contribution in (\ref{g2}), we proceed as follows. By definition 
\be
R^{*}_{t} A(f,g) R_{t} = A(\n_{N,t}(f,g))\label{g2b}
\ee
for all $f,g\in \h \oplus \h$. Here $\nu_{N,t}$ is the Bogoliubov transformation
\be
\nu_{N,t} = \begin{pmatrix} U_{N,t} & \overline V_{N,t} \\ V_{N,t} & \overline U_{N,t} \end{pmatrix}
\ee
with $U_{N,t}$, $V_{N,t}$ given by
\be
U_{N,t} = \begin{pmatrix}  u_{N,t} & 0 \\ 0 & \overline{u_{N,t}} \end{pmatrix}\;,\qquad V_{N,t} = \begin{pmatrix} 0 & \overline{v_{N,t}} \\ -v_{N,t} & 0 \end{pmatrix} \label{g3}
\ee
and $u_{N,t} = \sqrt{1 - \o_{N,t}}$, $v_{N,t} = \sqrt{\o_{N,t}}$. Differentiating the left hand side of (\ref{g2b}) we get
\bea
i\e \partial_{t} \left(R^{*}_{t} A(f,g) R_{t}\right) &=& (i\e \partial_{t} R^{*}_{t}) A(f,g) R_{t} + R^{*}_{t} A(f,g) i\e \partial_{t} R_{t}\nn\\
&=& (i\e \partial_{t} R^{*}_{t})R_{t} R^{*}_{t} A(f,g) R_{t} + R^{*}_{t} A(f,g) R_{t} R^{*}_{t} i\e \partial_{t} R_{t}\nn\\
&=& (i\e \partial_{t} R^{*}_{t})R_{t} A(\nu_{N,t}(f,g)) + A(\nu_{N,t}(f,g)) R^{*}_{t} i\e \partial_{t} R_{t}\nn\\
&=& (i\e \partial_{t} R^{*}_{t})R_{t} A(\nu_{N,t}(f,g)) - A(\nu_{N,t}(f,g)) (i\e \partial_{t} R^{*}_{t})R_{t}\nn\\
&=& \Big[ (i\e \partial_{t} R^{*}_{t})R_{t} , A(\nu_{N,t}(f,g))\Big] \label{g4}
\eea
where in the fourth line we used that $0 = \partial_{t}\left( R^{*}_{t}R_{t}\right) = (\partial_{t}R^{*}_{t})R_{t} + R^{*}_{t}\partial_{t}R_{t}$. Therefore, from (\ref{g2b}), (\ref{g4}) we find
\be
\Big[ (i\e \partial_{t} R^{*}_{t})R_{t} , A(\nu_{N,t}(f,g))\Big] = i\e A(\partial_{t}\n_{N,t}(f,g)).\label{g5}
\ee
Since (\ref{g5}) holds true for all $f,g\in \h \oplus \h$, $(i\e \partial_{t} R^{*}_{t})R_{t}$ must be quadratic in the fermionic operators,
\bea
(i\e \partial_{t} R^{*}_{t})R_{t} &=& \sum_{\s,\s' = l,r}\Big[\int dxdy\, C_{t}(x,\s;y,\s')a^{*}_{x,\s}a_{y,\s'} \nn\\
&& + \frac{1}{2}\Big(\int dxdy\, D_{t}(x,\s;y,\s')a^{*}_{x,\s}a^{*}_{y,\s'} + h.c.\Big)\Big]\nn\\
&\equiv& \int d\xx d\yy\,C_{t}(\xx;\yy)a^{*}_{\xx}a_{\yy} + \frac{1}{2}\Big(\int d\xx d\yy\, D_{t}(\xx;\yy)a^{*}_{\xx}a^{*}_{\yy} + h.c.\Big)\label{g6}
\eea
where we used the shorthand notations $\xx=(x,\s)$, $\int d\xx = \sum_{\s=l,r}\int dx$. Eq.\ (\ref{g6}) holds for suitable kernels $C_{t}$, $D_{t}$, to be determined later, such that
\be
\overline{C_{t}(\xx;\yy)} = C_{t}(\yy;\xx)\;,\qquad D_{t}(\xx;\yy) = -D_{t}(\yy;\xx)\;.
\ee
Eq.\ (\ref{g5}) is the equation that determines $C_{t}$, $D_{t}$. Taking $f(x) = (\delta(x-w)\delta_{l,\s''}, \delta(x-w)\delta_{r,\s''})$, $g(x)=(0, 0)$, and setting $\ww = (w, \s'')$, we find
\[
A(\nu_{N,t}(f,g)) = A(U_{N,t}f, V_{N,t}f) = a(U_{N,t}f) + a^{*}(\overline{V}_{N,t}\overline{f}) = a(U_{t,\ww}) + a^{*}(\overline{V}_{t,\ww})\;
\]
where we introduced the notation
\[
a(U_{t,\ww}) := \sum_{\s}\int dz\, a_{z,\s} \overline{U_{N,t}(z,\s; w,\s'')},\qquad a^{*}(\overline{V}_{t,\ww}) := \sum_{\s}\int dz\, a^{*}_{z,\s} \overline{V_{N,t}(z,\s; w,\s'')}.
\]
Thus, the left hand side of (\ref{g5}) can be rewritten as:
\bea
&&\Big[ (i\e \partial_{t} R^{*}_{t})R_{t}, a(U_{t,\ww}) + a^{*}(\overline{V_{t,\ww}})\Big] \nn\\
&& \quad= \int d\xx d\yy\,C_{t}(\xx;\yy) \Big[ a^{*}_{\xx}a_{\yy}, a(U_{t,\ww}) + a^{*}(\overline{V_{t,\ww}}) \Big]  \nn\\
&& \quad\quad + \frac{1}{2}\int d\xx d\yy\, D_{t}(\xx;\yy) \Big[ a^{*}_{\xx}a^{*}_{\yy}, a(U_{t,\ww}) + a^{*}(\overline{V_{t,\ww}}) \Big] \nn\\
&& \quad\quad + \frac{1}{2}\int d\xx d\yy\, \overline{D_{t}(\xx;\yy)} \Big[ a_{\yy}a_{\xx}, a(U_{t,\ww}) + a^{*}(\overline{V_{t,\ww}}) \Big].\label{g7}
\eea
Let us compute the first term on the r.h.s. of (\ref{g7}). We get 
\bea
&&\int d\xx d\yy\, C_{t}(\xx;\yy)\Big[ a^{*}_{\xx}a_{\yy}, a(U_{t,\ww}) + a^{*}(\overline{V_{t,\ww}}) \Big] \nn\\
&& \quad = \int d\xx d\yy\, C_{t}(\xx;\yy) (-a_{\yy}\overline{{U}_{t}(\xx;\ww)} + a^{*}_{\xx}\overline{{V}_{t}(\yy;\ww)})\nn\\
&&\quad = -a(C_{t}U_{t,\ww}) + a^{*}(C_{t}\overline{V}_{t,\ww}).\label{g8}
\eea
Similarly, we find that
\bea
&&\frac{1}{2}\int d\xx d\yy\, D_{t}(\xx;\yy)\Big[ a^{*}_{\xx}a^{*}_{\yy}, a(U_{t,\ww}) + a^{*}(\overline{V}_{t,\ww}) \Big] \nn\\
&&\quad = \int d\xx d\yy\, D_{t}(\xx;\yy)(a^{*}_{\xx}\overline{U}_{t}(\yy;\ww) - a^{*}_{\yy}\overline{U}_{t}(\xx;\ww))\nn\\
&& \quad = a^{*}(D_{t}\overline{U}_{t,\ww})\label{g9}
\eea
and 
\be
\frac{1}{2}\int d\xx d\yy\, \overline{D_{t}(\xx;\yy)}\Big[ a_{\yy}a_{\xx}, a(U_{t,\ww}) + a^{*}(\overline{V}_{t,\ww}) \Big] = -a(D_{t}V_{t,\ww}).\label{g10}
\ee
Therefore, summing up (\ref{g8}) -- (\ref{g10}) and substituting  them into (\ref{g5}), we obtain
\be
a^{*}(C_{t}\overline{V}_{t,\ww}) + a^{*}(D_{t}\overline{U}_{t,\ww}) - a(C_{t}U_{t,\ww}) - a(D_{t}V_{t,\ww}) = a^{*}(i\e\partial_{t}\overline{V}_{t,\ww}) - a(i\e\partial_{t}U_{t,\ww}).
\ee
This translates into a system of equations for the kernels $C_{t}$ and $D_{t}$, namely:
\bea
C_{t}\overline{V}_{t} + D_{t}\overline{U}_{t} &=& i\e\partial_{t} \overline{V}_{t},\nn\\
C_{t}U_{t} + D_{t}V_{t} &=& i\e\partial_{t} U_{t},
\eea
which can be rewritten more compactly as 
\be
\begin{pmatrix} C_{t} & D_{t} \end{pmatrix} \begin{pmatrix} U_{t} & \overline{V}_{t} \\ V_{t} & \overline{U}_{t} \end{pmatrix} = \begin{pmatrix} i\e\partial_{t} U_{t} & i\e\partial_{t}\overline{V}_{t} \end{pmatrix}.
\ee
This gives 
\be
\begin{pmatrix} C_{t} & D_{t} \end{pmatrix} = \begin{pmatrix} i\e\partial_{t} U_{t} & i\e\partial_{t}\overline{V}_{t} \end{pmatrix}\begin{pmatrix} U_{t} & V^{*}_{t} \\ \overline{V}^{*}_{t} & \overline{U}_{t} \end{pmatrix}.\label{g11}
\ee
Let
\bea
H_{HF}(t) &:=& \begin{pmatrix} h_{HF}(t) & 0 \\ 0 & -\overline{h_{HF}(t)} \end{pmatrix},\nn\\
h_{HF}(t) &:=& -\e^{2}\D + \r_{t}* V - X_{t}.\nn
\eea
The time-dependent Hartree-Fock equation $i\e\partial_{t}\o_{t} = [h_{HF}(t),\o_{t}]$ implies
\bea
i\e \partial_{t} U_{t} &=& [H_{HF}(t), U_{t}].\nn\\
i\e \partial_{t} \overline{V}_{t} &=& H_{HF} (t)\overline{V_{t}} + \overline{V_{t}}\overline{H_{HF}(t)}.
\eea
Inserting these equations into (\ref{g11}) we have 
\bea
C_{t} &=& H_{HF}(t) - U_{t}H_{HF}(t)U_{t} + \overline{V}_{t}\overline{H_{HF}(t)}\overline{V_{t}}^{*}\nn\\
&=&\begin{pmatrix} h_{HF} (t) - u_{t}h_{HF} (t) u_{t} - v_{t}h_{HF} (t) v_{t} & 0 \\ 0 & -\overline{h_{HF} (t)}  + \overline{u_{t}} \overline{h_{HF} (t)}\overline{u_{t}} + \overline{v_{t}}\overline{h_{HF} (t)}\overline{v_{t}} \end{pmatrix}, \nn\\
D_{t} &=& -U_{t} H_{HF}(t) V^{*}_{t} + \overline{V}_{t} \overline{H_{HF}(t)}\overline{U_{t}}\nn\\
&=& \begin{pmatrix} 0 & u_{t}h_{HF} (t) v_{t} - v_{t}h_{HF} (t) u_{t} \\ \overline{u_{t}}\overline{h_{HF} (t)}\overline{v_{t}} - \overline{v_{t}}\overline{h_{HF} (t)}\overline{u_{t}} & 0 \end{pmatrix}.
\eea
Therefore, we conclude that
\begin{equation}\label{eq:idtRR}
\begin{split}
(i\e\partial_{t} R^{*}_{t})R_{t} = \; & d\Gamma_{l}\Big(h_{HF} (t) - u_{t}h_{HF} (t) u_{t} - v_{t}h_{HF} (t) v_{t}\Big) \\
& - d\Gamma_{r}\Big(\overline{h_{HF} (t)} - \overline{u_{t}}\overline{h_{HF} (t)}\overline{u_{t}} - \overline{v_{t}}\overline{h_{HF} (t)}\overline{v_{t}}\Big)\\
&+ \Big(\int dxdy\, a^{*}_{x,l}a^{*}_{y,r}\big(u_{t}h_{HF} (t) v_{t} - v_{t}h_{HF} (t) u_{t}\big)(x,y) + h.c.\Big). 
\end{split}
\end{equation}
\medskip
{\it Computation of $R^{*}_{t}\cL_{N}R_{t}$.} Let us first compute the contribution of the kinetic energy, that is
\be
R^{*}_{t}d\Gamma_{l}(-\e^{2}\D)R_{t} - R^{*}_{t}d\Gamma_{r}(-\e^{2}\D)R_{t}.
\ee
Using the definition of Bogoliubov transformation, after a straightforward computation we have 
\bea
R^{*}_{t}d\Gamma_{l}(-\e^{2}\D)R_{t} &=& \Tr\,(-\e^{2}\D\o) + d\Gamma_{l}(u_{t}(-\e^{2}\D)u_{t}) - d\Gamma_{r}(\overline{v_{t}}(-\e^{2}\D)\overline{v_{t}})\nn\\
&& - \Big( \int dxdy\, a^{*}_{x,l}a^{*}_{y,r}\big( u_{t}(-\e^{2}\D)v_{t} \big)(x,y) + h.c. \Big)
\eea
and, analogously,  
\bea
R^{*}_{t}d\Gamma_{r}(-\e^{2}\D)R_{t} &=& \Tr\,(-\e^{2}\D \overline\o) - d\Gamma_{l}(v_{t}(-\e^{2}\D)v_{t}) + d\Gamma_{r}(\overline{u_{t}}(-\e^{2}\D)\overline{u_{t}})\nn\\
&& + \Big( \int dxdy\, a^{*}_{x,r}a^{*}_{y,l}\big( \overline{u_{t}}(-\e^{2}\D)\overline{v_{t}} \big)(x,y) + h.c. \Big).
\eea
In other words, 
\bea
&&R^{*}_{t} \left[ d\Gamma_{l}(-\e^{2}\D) - d\Gamma_{r}(-\e^{2}\D) \right] R_{t} = \nn\\
&& \quad = d\Gamma_{l}\Big(u_{t}(-\e^{2}\D)u_{t} + v_{t}(-\e^{2}\D)v_{t}\Big)\nn\\
&& \quad\quad - d\Gamma_{r}\Big(\overline{u_{t}}(-\e^{2}\D)\overline{u_{t}} + \overline{v_{t}}(-\e^{2}\D)\overline{v_{t}}\Big)\nn\\
&& \quad\quad - \Big(\int dxdy\, a^{*}_{x,l}a^{*}_{y,r}\big( u_{t}(-\e^{2}\D)v_{t} - v_{t}(-\e^{2}\D)u_{t} \big)(x,y) + h.c.\Big).\label{eq:327}
\eea
Finally, we compute the contribution coming from the many-body interaction, namely 
\be
R^{*}_{t}\Big(\frac{1}{2N}\int dxdy\, V(x-y)(a^{*}_{x,l}a^{*}_{y,l}a_{y,l}a_{x,l} -a^{*}_{x,r}a^{*}_{y,r}a_{y,r}a_{x,r}) \Big)R_{t}.\label{g11b}
\ee
Consider first the ``left'' contribution to (\ref{g11b}). Putting all terms in normal order and recalling the shorthand notation $u_x (y) := u_{N,t} (y,x)$ and $v_x (y) := v_{N,t} (y,x)$ we find 
\begin{equation}
\begin{split} 
& R^{*}_{t} a^{*}_{x,l} a^{*}_{y,l}a_{y,l}a_{x,l}R_{t} \\
&= (a^{*}_{l}(u_{x}) - a_{r}(\overline{v_{x}}) )(a^{*}_{l}(u_{y}) - a_{r}(\overline{v_{y}}) )(a_{l}(u_{y}) - a^{*}_{r}(\overline{v_{y}}) )(a_{l}(u_{x}) - a^{*}_{r}(\overline{v_{x}}) )  \\
&=  \big[ a^{*}_{l}(u_{x})a^{*}_{l}(u_{y})a^{*}_{r}(\overline{v_{y}})a^{*}_{r}(\overline{v_{x}}) + 2a^{*}_{l}(u_{x})a^{*}_{l}(u_{y})a^{*}_{r}(\overline{v_{x}})a_{l}(u_{y}) - 2a^{*}_{l}(u_{x})a^{*}_{r}(\overline{v_{y}})a^{*}_{r}(\overline{v_{x}})a_{r}(\overline{v_{y}})\\
&\quad + \frac{1}{2}a^{*}_{l}(u_{x})a^{*}_{l}(u_{y})a_{l}(u_{y})a_{l}(u_{x}) + a^{*}_{l}(u_{x})a^{*}_{r}(\overline{v_{x}})a_{r}(\overline{v_{y}})a_{l}(u_{y}) - a^{*}_{l}(u_{x})a^{*}_{r}(\overline{v_{y}})a_{r}(\overline{v_{y}})a_{l}(u_{y})\\
&\quad + \frac{1}{2}a^{*}_{r}(\overline{v_{y}})a^{*}_{r}(\overline{v_{x}})a_{r}(\overline{v_{x}})a_{r}(\overline{v_{y}})\big] + h.c. + \text{Q} + \text{S}
\end{split} \end{equation}
where 
\begin{equation}\label{eq:quad}
\begin{split}
\text{Q} = \; & \big[-2a^{*}_{l}(u_{x})a^{*}_{r}(\overline{v_{x}})\o(y,y) + 2a^{*}_{l}(u_{x})a^{*}_{r}(\overline{v_{y}})\o(x,y)\big] + h.c.\\
& + 2a^{*}_{l}(u_{y})a_{l}(u_{y})\o(x,x) - 2a^{*}_{l}(u_{y})a_{l}(u_{x})\o(y,x) - 2a^{*}_{r}(\overline{v_{y}})a_{r}(\overline{v_{y}})\o(x,x)
\\ & + 2a^{*}_{r}(\overline{v_{y}})a_{r}(\overline{v_{x}})\o(x,y)
\end{split} 
\end{equation}
are the quadratic terms and 
\be
\text{S} = \o(x,x)\o(y,y) - \o(y,x)\o(x,y)
\ee
are the scalar terms (both the quadratic and the scalar terms arise from the normal ordering procedure). To compute the ``right'' contribution to (\ref{g11b}), one just has to perform the replacement $l\leftrightarrow r$, $u_{x}\to \overline{u_{x}}$, $\overline{v_{x}}\to -v_{x}$ in the above computation (see Eq.\ (\ref{a8b})). 

Combining (\ref{eq:idtRR}), (\ref{eq:327}), (\ref{eq:quad}) and its ``right'' counterpart, we conclude that 
\bea
\mathcal{G}_{N}(t) &=& (i\e\partial_{t}R^{*}_{t})R_{t} + R^{*}_{t}\cL_{N}R_{t}\nn\\
&=& d\Gamma_{l}(h_{HF}(t)) - d\Gamma_{r}(\overline{h_{HF}(t)}) + \mathcal{C_{N}} + \mathcal{Q}_{N}
\eea
where $\mathcal{C}_N$ and $\mathcal{Q}_N$ are  given by  (\ref{eq:CN}) and (\ref{eq:QN}). 
\end{proof}

\section{Bounds on the growth of fluctuations}\label{sec:growth}
\setcounter{equation}{0}

The goal of this section is to control the growth of the particle number  operator with respect to the fluctuation dynamics. 
\begin{thm}\label{thm:maingrowth}
Assume $V \in L^1 (\bR^3)$ satisfies (\ref{eq:Vass}). Let $\omega_N$ be  a sequence of trace class operators on $\h = L^2 (\bR^3)$ with $0 \leq \omega_N \leq 1$, $\tr(1-\Delta)\omega_{N}<\infty$, $\tr \omega_{N} = N$, and satisfying the semiclassical bounds (\ref{eq:oass}). Let $\xi\in \cF (\h \oplus \h)$ be such that $\xi = \chi(\cN\leq CN)\xi$, where $C\geq 0$ is independent of $N$. Let $\cU_N (t;s)$ be the fluctuation dynamics, as defined in (\ref{eq:flu-dyn}). Let $k \in \bN$. Then there exists a constant $c_1$ depending on $V$ and on $k$ and a constant $c_2$ depending only on $V$ such that
\be
\Big\langle \cU_{N}(t;0)\xi, (\cN+1)^{k} \cU_{N}(t;0)\xi \Big\rangle \leq \exp(c_{1}\exp(c_{2}|t|)) \Big\langle \xi, (\cN+1)^{10k} \xi \Big\rangle
\ee
for all $t\in\mathbb{R}$.
\end{thm}

The proof of this theorem relies on Propositions~\ref{thm:growth} and~\ref{prop:aux} that will be proven in the next two sections. In turn, these results rely on the estimates for operators on the extended Fock space $\cF (\h \oplus \h)$ contained in the following lemma (which is the analog of Lemma 3.1 in \cite{BPS}).

\begin{lem}\label{lem:estimates}
For every bounded operator $O$ on $\h = L^2 (\bR^3)$, we have
\be\label{est1}
\| d\G_{\s}(O)\psi \| \leq \| O \|\|\cN\psi\|
\ee
for every $\psi\in\cF(\h\oplus \h)$ and every $\s \in \{l, r\}$. Here $\cN = d\G_l (1) + d\G_r (1) = d\G (1\oplus 1)$ denotes the particle number operator on $\cF (\h \oplus \h)$. 

If $O$ is a Hilbert-Schmidt operator, we also have 
\be\label{est2}
\begin{split}
\| d\G_{\s}(O)\psi \| &\leq \| O \|_{HS}\| \cN^{1/2}\psi \|, \\
\Big\| \int dxdx'\, O(x;x')a_{x,\s}a_{x',\s'}\psi \Big\| &\leq \| O \|_{HS}\| \cN^{1/2}\psi \|, \\
\Big\| \int dxdx'\, O(x;x')a^{*}_{x,\s}a^{*}_{x',\s'}\psi \Big\| &\leq 2\| O \|_{HS}\| (\cN + 1)^{1/2}\psi \|
\end{split}
\ee
for all $\psi\in \cF(\h\oplus \h)$.
\end{lem}

\begin{proof}
To prove (\ref{est1}) notice that, for bounded operators $A,B$ on $\h$,
\be\label{est1a}
\begin{split}
\| d\G(A\oplus B)\psi \|^{2} &= \sum_{n\geq 0}\sum_{i,j\leq n}\langle \psi^{(n)}, (A\oplus B)^{(i)*}(A\oplus B)^{(j)}\psi^{(n)} \rangle \\
&\leq  \| A\oplus B \|_{\h \oplus \h}^{2}\sum_{n\geq 0}n^{2}\|\psi^{(n)}\|^{2} \\
&\leq \| A\oplus B \|_{\h \oplus \h}^{2}\|\cN\psi\|^{2}.
\end{split}
\ee
The proof of (\ref{est1}) follows from (\ref{est1a}) and the inequality   $\| A\oplus B \|_{\h \oplus \h}\leq \|A\|_{\h} + \|B\|_{\h}$. Let us now prove the bounds (\ref{est2}). To prove the first two, notice that, denoting by $a^{\sharp}$ either $a$ or $a^{*}$ and by $O^{\sharp}$ either $\overline O$ or $O$: 
\be
\begin{split}
\Big\| \int dxdx'\, O(x;x')a^{\sharp}_{x,\s}a_{x',\s'}\psi \Big\| &= \Big\| \int dx' a^{\sharp}_{\s}(O^{\sharp}(\cdot;x'))a_{x',\s'}\psi \Big\| \\
&\leq \int dx'\, \| O(\cdot;x') \|_{2}\| a_{x',\s'}\psi \| \\
&\leq \big(\int dx'\, \| O(\cdot;x') \|_{2}^{2}\big)^{1/2}\big(\int dx'\, \| a_{x',\s'}\psi \|^{2}\big)^{1/2} \\
&\leq \| O \|_{HS}\| \cN^{1/2}_\s \psi \| \leq \| O \|_{HS}\| \cN^{1/2} \psi \|
\end{split}
\ee
where in the second line we used the boundedness of the fermionic operators, Eqs. (\ref{eq:bdaa*}). This proves in particular that
\[ \left\| \int dx dx' O(x;x') a_{x,\s} a_{x',\s'} (\cN+1)^{-1/2} \right\| \leq \| O \|_{HS}. \]
Since the adjoint of a bounded operator is again bounded and has the same operator norm, we find
\[ 
\begin{split}
\left\| \int dx dx' O (x;x') (\cN+1)^{-1/2} a^*_{x,\s} a^*_{x',\s'} \right\| &= \left\| \int dx dx' O (x;x')  a^*_{x,\s} a^*_{x',\s'} (\cN+3)^{-1/2}\right\| \\
& \leq \| O \|_{HS}.
\end{split}
\]
This concludes the proof, since $\| (\cN+3)^{1/2} \psi\| \leq 2 \| (\cN+1)^{1/2} \psi \|$ for all $\psi \in \cF (\h \oplus \h)$.
\end{proof}

\subsection{The auxiliary dynamics}\label{sec:auxiliary}

We start by studying the growth of the number of fluctuations with respect to an {\it auxiliary dynamics} that we shall denote by $\widetilde{\cU}_{N}(t;s)$. Later we will prove that $\wt{\cU}_N (t;s)$ remains close (in norm) to the true fluctuation dynamics $\cU_{N}(t;s)$.

\medskip

To define the auxiliary dynamics, we proceed as follows. First, let
\begin{equation}\label{f1b}
\widetilde{\mathcal{G}}_{N}(t) = 
 d\Gamma_{l}(h_{HF}(t)) - d\Gamma_{r}(\overline{h_{HF}(t)}) + \mathcal{C}_{N} + \widetilde{\mathcal{Q}}_{N} 
\end{equation}
where 
\begin{equation}\label{eq:qtilde}
\begin{split}
\widetilde{\mathcal{Q}}_{N} = & \; \frac{1}{2N}\int dxdy\, V(x-y)\\& \hspace{1cm} \times \Big(a^{*}_{l}(u_{x})a^{*}_{l}(u_{y})a^{*}_{r}(\overline{v_{y}})a^{*}_{r}(\overline{v_{x}}) - a^{*}_{r}(\overline{u_{x}})a^{*}_{r}(\overline{u_{y}})a^{*}_{l}(v_{y})a^{*}_{l}(v_{x}) + h.c.\Big).
\end{split} 
\end{equation}
In (\ref{f1b}), the operator $\mathcal{C}_N$ is defined as in (\ref{eq:CN}) and contains quartic terms commuting with the particle number operator $\cN = d\G_l (1) + d\G_r (1)$. In other words, $\widetilde{\mathcal{G}}_{N}(t)$ is obtained 
from $\mathcal{G}_{N}(t)$ by dropping the second, third, fifth and sixth terms in (\ref{eq:QN}) (together with their Hermitian conjugates). We recall the shorthand notations $u_{x}(y) = u_{N,t}(y,x)$, $v_{x}(y) = v_{N,t}(y,x)$.

We are interested in the dynamics generated by $\widetilde{\mathcal{G}}_{N}(t)$. For technical reasons, it is more convenient to define it in the interaction picture. Thus, we introduce
\be\label{eq:ghat}
\widehat{\mathcal{G}}_{N}(t) = -d\Gamma_{l}(-\e^{2}\Delta) + d\Gamma_{r}(-\e^{2}\Delta) + \cU_{N}^{(0)*}(t)\widetilde{\mathcal{G}}_{N}(t)\cU_{N}^{(0)}(t)\;,
\ee
where $\cU_{N}^{(0)}(t)$ implements the free dynamics, {\it i.e.} it is given by the solution of
\be
i\e\frac{d}{dt}\cU_{N}^{(0)}(t) = \Big(d\Gamma_{l}(-\e^{2}\Delta) - d\Gamma_{r}(-\e^{2}\Delta)\Big)\cU^{(0)}_{N}(t)\;,\qquad \cU_{N}^{(0)}(0) = 1\;.
\ee
We claim that $\widehat{\mathcal{G}}_{N}(t)$ generates a unitary dynamics $\widehat{\cU}_{N}(t;s)$ on the Fock space, satisfying the differential equation:
\be\label{eq:Uhat}
i\e \frac{d}{dt} \widehat{\cU}_{N}(t;s)\psi = \widehat{\mathcal{G}}_{N}(t)\widehat{\cU}_{N}(t;s)\psi\;,\qquad \widehat{\cU}_{N}(s;s)\psi = \psi\;,
\ee
for $\psi \in D(\cN^{4})$, the domain of $\cN^{4}$. We prove these facts in Appendix \ref{app:derivative}. Then, we define the auxiliary dynamics as 
\be\label{f1}
\widetilde{\cU}_{N}(t;s) := \cU_{N}^{(0)}(t)\widehat{\cU}_{N}(t;s)\cU_{N}^{(0)*}(s).
\ee
Note that, formally, $\widetilde{\cU}_{N}(t;s)$ is the unitary dynamics generated by $\widetilde{\mathcal{G}}_{N}(t)$. An important role in our analysis is played be the expectation of $\cN^{k}$ on states $\widetilde{\cU}_{N}(t;s)\psi$, for $\psi$ in the domain of a suitable power of $\cN$. To estimate these quantities, we shall set up a Gronwall-type strategy. The starting point is:
\[
\begin{split}\label{eq:Utilde}
i\e\frac{d}{dt}\Big\langle \widetilde{\cU}_{N}(t;s)\psi, (\cN+1)^{k} &\widetilde{\cU}_{N}(t;s)\psi \Big\rangle \\
&= i\e\frac{d}{dt}\Big\langle \widehat{\cU}_{N}(t;s)\cU^{(0)*}_{N}(s)\psi, (\cN+1)^{k} \widehat{\cU}_{N}(t;s)\cU^{(0)*}_{N}(s)\psi \Big\rangle \\
&=  \Big\langle \widehat{\cU}_{N}(t;s)\cU^{(0)*}_{N}(s)\psi, [(\cN+1)^{k}, \widehat{\mathcal{G}}_{N}(t)]\, \widehat{\cU}_{N}(t;s)\cU^{(0)*}_{N}(s)\psi \Big\rangle \\
&= \Big\langle \widetilde{\cU}_{N}(t;s)\psi, [(\cN+1)^{k}, \widetilde{\mathcal{Q}}_{N}]\, \widetilde{\cU}_{N}(t;s)\psi \Big\rangle
\end{split}
\]
where we used Eq. (\ref{eq:Uhat}), and the fact that $d\Gamma_{\sigma}(-\e^{2}\Delta)$, $\s=l,r$, and $\mathcal{C}_{N}$ commute with $\cN^{k}$. With the following proposition, we prove a bound for the growth of the fluctuations with respect to $\widetilde{\cU}_{N}(t;s)$.
\begin{prp}\label{thm:growth}
Let $V \in L^1 (\bR^3)$ and suppose that  (\ref{eq:Vass}) holds. Let $\omega_N$ be a sequence of trace-class operators on $\h = L^{2}(\mathbb{R}^{3})$ with $0 \leq \omega_N \leq 1$, $\tr (1-\Delta)\o_{N}<\infty$, $\tr \omega_N = N$ and satisfying (\ref{eq:oass}). Let $\widetilde{\cU}_{N}(t;s)$ be the auxiliary dynamics defined in (\ref{f1}) and $k\in\mathbb{N}$. Then there exist a constant $c_{1}$, depending only on $V$, and a constant $c_{2}$, depending on $V$ and $k$, such that
\be
\Big\langle \widetilde{\cU}_{N}(t;s)\xi, (\cN+1)^{k} \widetilde{\cU}_{N}(t;s)\xi \Big\rangle \leq \exp(c_{2}\exp(c_{1}|t-s|)) \Big\langle \xi, (\cN+1)^{k} \xi \Big\rangle.
\ee
\end{prp}
As in \cite{BPS, BPS2}, the proof of the theorem makes use of the propagation of the commutator estimates (\ref{eq:oass})   along the solution of the Hartree-Fock equation. We defer the proof of the following Lemma to Appendix \ref{secpropagation}.
\begin{lem}\label{lem:prop}
Let $V \in L^1 (\bR^3)$, and suppose that  (\ref{eq:Vass}) holds. Let $\omega_N$ be a sequence of trace-class operators on $\h$ with $0 \leq \omega_N \leq 1$, $\tr (1-\Delta)\o_{N}<\infty$, $\tr \omega_N = N$ and such that
\be
\begin{split}
\|[v_N, x]\|_{HS} &\leq CN^{1/2}\e, \qquad \|[v_N,\e\nabla]\|_{HS} \leq CN^{1/2}\e, \\
\|[u_{N}, x]\|_{HS} &\leq CN^{1/2}\e, \qquad  \|[u_N,\e\nabla]\|_{HS} \leq CN^{1/2}\e,\label{f2}
\end{split}
\ee
where $v_{N} = \sqrt{\o_{N}}$, $u_{N} = \sqrt{1 - \o_{N}}$. Let $v_{N,t} = \sqrt{\omega_{N,t}}$ and $u_{N,t} = \sqrt{1-\omega_{N,t}}$, where $\omega_{N,t}$ denotes the solution of the Hartree-Fock equation with initial data $\omega_{N,0} = \omega_N$. 

Then, there exist constants $K,c>0$ only depending on the potential $V$ such that 
\be
\begin{split}
\|[v_{N,t}, x]\|_{HS} &\leq K\exp(c|t|)N^{1/2}\e, \qquad \|[v_{N,t},\e\nabla]\|_{HS} \leq K\exp(c|t|)N^{1/2}\e, \\
\|[u_{N,t}, x]\|_{HS} &\leq K\exp(c|t|)N^{1/2}\e, \qquad  \|[u_{N,t},\e\nabla]\|_{HS} \leq K\exp(c|t|)N^{1/2}\e,\label{f2b}
\end{split}
\ee
for all $t \in \bR$. 
\end{lem}

We are now ready to prove Proposition \ref{thm:growth}.
\medskip

\noindent{\it Proof of Proposition \ref{thm:growth}.} For simplicity, we set $s=0$; the case $s\neq 0$ can be handled in a completely analogous way. Using (\ref{eq:Utilde}), we compute
\begin{equation}\label{f2c}
\begin{split} 
i\e & \frac{d}{dt}  \Big\langle \xi, \widetilde{\cU}_N^* (t;0) (\cN+1)^k \widetilde{\cU}_N (t;0) \xi \Big\rangle \\ & = \Big\langle \xi, \widetilde{\cU}_N^* (t;0) \left[ (\cN+1)^k , \widetilde{\mathcal{Q}}_{N} \right] 
\widetilde{\cU}_N (t;0) \xi \Big\rangle \\
& = \sum_{j=1}^k  \Big\langle \xi, \widetilde{\cU}_N^* (t;0) (\cN+1)^{j-1} \left[ \cN  , \widetilde{\mathcal{Q}}_{N} \right] 
(\cN+1)^{k-j} \, \widetilde{\cU}_N (t;0) \xi \Big\rangle.
\end{split} 
\end{equation}
We have
\[\begin{split}  \left[ \cN,  \widetilde{\mathcal{Q}}_{N} \right] &= - \frac{4i}{N} \text{Im }  \int dxdy\, V(x-y)\, \\ &\hspace{2cm} \times \Big\{ a^{*}_{l}(u_{x})a^{*}_{l}(u_{y})a^{*}_{r}(\overline{v_{y}})a^{*}_{r}(\overline{v_{x}}) - a^{*}_{r}(\overline{u_{x}})a^{*}_{r}(\overline{u_{y}})a^{*}_{l}(v_{y})a^{*}_{l}(v_{x}) \Big\}.
\end{split}
\]
It is important to recognize a cancellation arising from the two terms in the parenthesis. We rewrite the commutator as
\be\label{f4}
\begin{split}
\left[ \cN , \widetilde{\mathcal{Q}}_{N} \right] = \; & -\frac{4i}{N} \text{Im } \int dp \widehat{V} (p) \int dz_{1}dz_{2}dz_{3}dz_{4}\\ &\hspace{1cm} \times  \Big[ a^{*}_{z_{1},l}a^{*}_{z_{2},l}a^{*}_{z_{3},r}a^{*}_{z_{4},r}(u_t e^{ipx}v_t )(z_{1},z_{4})(u_t e^{-ipx}v_t)(z_{2},z_{3}) \\
& \hspace{.5cm} - a^{*}_{z_{1},r}a^{*}_{z_{2},r}a^{*}_{z_{3},l}a^{*}_{z_{4},l}(\overline{u}_t e^{ipx}\overline{v}_t)(z_{1},z_{4})(\overline{u}_t e^{-ipx}\overline{v}_t)(z_{2},z_{3})\Big) \\
= \; &- \frac{4i}{N} \text{Im } \int dp\,\widehat{V} (p)\int dz_{1}dz_{2}dz_{3}dz_{4}\, a^{*}_{z_{1},l}a^{*}_{z_{2},l}a^{*}_{z_{3},r}a^{*}_{z_{4},r} \\ 
& \times \Big[ (u_t e^{ipx} v_t)(z_{1},z_{4})(u_t e^{-ipx} v_t)(z_{2},z_{3}) - (\overline{u}_t e^{ipx}\overline{v}_t )(z_{4},z_{1})(\overline{u}_t e^{-ipx}\overline{v}_t )(z_{3},z_{2}) \Big]  \\
= \; &- \frac{4i}{N} \text{Im } \int dp\,\hat V(p)\int dz_{1}dz_{2}dz_{3}dz_{4}\, a^{*}_{z_{1},l}a^{*}_{z_{2},l}a^{*}_{z_{3},r}a^{*}_{z_{4},r} \\
&\times \Big[ (u_t e^{ipx} v_t )(z_{1},z_{4})(u_t e^{-ipx} v_t )(z_{2},z_{3}) - (v_t e^{ipx} u_t )(z_{1},z_{4})(v_t e^{-ipx} u_t )(z_{2},z_{3}) \Big]. 
\end{split}
\ee
Using the identity 
\be
u_t e^{ipx} v_t = v_t e^{ipx} u_t + u_t [e^{ipx},v_t] + v_t [u_t ,e^{ipx}]\label{f5}
\ee
we can rewrite (\ref{f4}) as
\be\label{f6}
\begin{split}
\left[ \cN , \widetilde{\mathcal{Q}}_{N} \right] = \; &- \frac{4i}{N} 
\text{Im } \int dp\,\widehat{V} (p)\int dz_{1}dz_{2}dz_{3}dz_{4}\, a^{*}_{z_{1},l}a^{*}_{z_{2},l}a^{*}_{z_{3},r}a^{*}_{z_{4},r} \\
&\hspace{.5cm} \times \Big[ (v_t e^{ipx}u_t)(z_{1},z_{4})\Big((u_t [e^{-ipx},v_t])(z_{2},z_{3}) + (v_t[u_t,e^{-ipx}])(z_{2},z_{3}) \Big) \\ &\hspace{1cm} + \Big( (u_t[e^{ipx},v_t])(z_{1},z_{4}) + (v_t[u_t,e^{ipx}])(z_{1},z_{4}) \Big)(u_t e^{-ipx} v_t)(z_{2},z_{3}) \Big].
\end{split}
\ee
Using Lemma \ref{lem:estimates} we find,  for any $\varphi\in \mathcal{F}$, 
\be\label{f7}
\begin{split}
\Big\| \int dz_{1}dz_{2}\, a^{*}_{z_{1},l}a^{*}_{z_{2},r}(u_t [e^{ipx}, v_t ])(z_{1},z_{2})\int &dz_{3}dz_{4}\, a^{*}_{z_{3},l}a^{*}_{z_{4},r} (u_t e^{-ipx}v_t)(z_{3},z_{4})\varphi \Big\| \\
& \leq C\| u_t [e^{ipx},v_t] \|_{HS} \| u_t e^{-ipx} v_t\|_{HS} \| (\cN + 1)\varphi \|\\
& \leq C N^{1/2} \| [e^{ipx}, v_t] \|_{HS} \| (\cN + 1)\varphi \|.
\end{split}
\ee
Hence, substituting  (\ref{f6}) into (\ref{f2c}), we find, after some simple manipulations, 
\begin{equation*}
\begin{split}
\eps \, \Big|
\frac{d}{dt}  & \Big\langle \xi, \widetilde{\cU}_N^* (t;0) (\cN+1)^k \widetilde{\cU}_N (t;0) \xi \Big\rangle \Big| \\ &\leq \frac{C}{N^{1/2}} \int dp | \widehat{V} (p)| \Big(\| [e^{ipx},v_t] \|_{HS} + \| [e^{ipx}, u_t] \|_{HS}\Big) \Big\langle \xi, \wt{\cU}_N^* (t;0) (\cN + 1)^{k} \wt{\cU}_N (t;0) \xi \Big\rangle 
\end{split} 
\end{equation*}
where the constant $C$ depends on $k$. Since
\be
\| [e^{ipx}, v_t] \|_{HS} \leq |p| \| [x, v_t] \|_{HS}\;,\qquad \| [e^{ipx}, u_t] \|_{HS} \leq |p|\| [x,u_t] \|_{HS}
\ee
we find, applying Proposition  \ref{lem:prop}, 
\be
\begin{split}
&\Big|\frac{d}{dt}  \Big\langle \xi, \widetilde{\cU}_N^* (t;0) (\cN+1)^k \widetilde{\cU}_N (t;0) \xi \Big\rangle \Big| \leq C \exp(c|t|) \Big\langle \xi, \widetilde{\cU}_N^* (t;0) (\cN+1)^k \widetilde{\cU}_N (t;0) \xi \Big\rangle.
\end{split}
\ee
Finally, Gronwall's Lemma implies that
\be
\Big\langle \xi, \widetilde{\cU}_N^* (t;0) (\cN+1)^k \widetilde{\cU}_N (t;0) \xi \Big\rangle \leq \exp(c_{2}\exp(c_{1}|t|))\Big\langle \xi, (\cN+1)^k \xi\Big\rangle
\ee
where $c_{2}$ depends on $k\in\mathbb{N}$ and on the potential $V$, and $c_{1}$ depends only on $V$. \qed

\subsection{Comparison of fluctuation and auxiliary dynamics}

In this section we prove that the auxiliary dynamics $\widetilde{\cU}_{N}$ is a good approximation of the fluctuation dynamics $\cU_{N}$, in the sense that $\| \cU_{N}(t;s)\xi - \widetilde{\cU}_{N}(t;s)\xi \|$ goes to zero as $N$ goes to infinity. This is the content of the following proposition.

\begin{prp}\label{prop:aux} Let $V \in L^1 (\bR^3)$ and suppose that  (\ref{eq:Vass}) holds. Let $\omega_N$ be a sequence of trace-class operators on $\h = L^{2}(\mathbb{R}^{3})$ with $0 \leq \omega_N \leq 1$, $\tr (1-\Delta)\o_{N}<\infty$, $\tr \omega_N = N$ and satisfying (\ref{eq:oass}). Let $\widetilde{\cU}_{N}(t;s)$ be the auxiliary dynamics defined in (\ref{f1}). Then there exist constants $c_{1}>0$, $c_{2}>0$, only depending on $V$, such that
\be
\| (\cU_{N}(t;s) - \widetilde{\cU}_{N}(t;s)) \xi \| \leq N^{-1/6}\exp(c_{2}\exp(c_{1}|t-s|)) \|(\cN+1)^{3/2} \xi \|.
\ee
\end{prp}

\begin{proof}
We write
\[
\begin{split}
i\e\partial_{t}\| (\cU_{N}(t;s) - \widetilde \cU_{N}(t;s))\xi \|^{2} &= i\e\partial_{t}\langle \O, (\cU_{N}(t;s)^{*} - \widetilde{\cU}_{N}(t;s)^{*})(\cU_{N}(t;s) - \widetilde{\cU}_{N}(t;s))\O\rangle\\
&= i\e\partial_{t}\langle \O, (2 - \cU_{N}(t;s)^{*}\widetilde \cU_{N}(t;s) - \widetilde{\cU}_{N}(t;s)^{*}\cU_{N}(t;s))\O\rangle\\
&= -2i\e\, \text{Im}\, \partial_{t} \langle \O, \cU_{N}(t;s)^{*}\widetilde \cU_{N}(t;s) \O \rangle\\
&= 2 \text{Im}\,\langle \cU_{N}(t;s)\xi, (\mathcal{G}_{N}(t) - \widetilde{\mathcal{G}}_{N}(t))\widetilde \cU_{N}(t;s)\xi \rangle\\
&= 2 \text{Im}\, \langle \cU_{N}(t;s) - \widetilde \cU_{N}(t;s)\xi, (\mathcal{G}_{N}(t) - \widetilde{\mathcal{G}}_{N}(t))\widetilde \cU_{N}(t;s)\xi \rangle 
\end{split}
\]
where in the last step we used that $\langle \widetilde \cU_{N}(t;s)\xi,(\mathcal{G}_{N}(t) - \widetilde{\mathcal{G}}_{N}(t))\widetilde \cU_{N}(t;s)\xi \rangle$ is real. This gives
\begin{equation} \e \left| \partial_t \| (\cU_{N}(t;s) - \widetilde \cU_{N}(t;s))\xi \| \right| \leq  \left\| (\mathcal{G}_{N}(t) - \widetilde{\mathcal{G}}_{N}(t))\widetilde \cU_{N}(t;s)\xi \right\|.\label{mix08} 
\end{equation}
Recall that 
\begin{equation}\label{eq:G-G} \begin{split} \mathcal{G}_{N}(t) -& \widetilde{\mathcal{G}}_{N}(t) \\ =  \; & \frac{1}{N}\int dxdy\, V(x-y)\Big(a^{*}_{l}(u_{x})a^{*}_{l}(u_{y})a^{*}_{r}(\overline{v_{x}})a_{l}(u_{y}) - a^{*}_{l}(u_{x})a^{*}_{r}(\overline{v_{y}})a^{*}_{r}(\overline{v_{x}})a_{r}(\overline{v_{y}}) \\
& \hspace{2cm} + a^{*}_{r}(\overline{u_{x}})a^{*}_{r}(\overline{u_{y}})a^{*}_{l}(v_{x})a_{r}(\overline{u_{y}}) - a_{r}^{*}(\overline{u_{x}})a^{*}_{l}(v_{y})a^{*}_{l}(v_{x})a_{l}(v_{y}) + h.c.\Big). \end{split} 
\end{equation} 
The contribution from the first term on the r.h.s. of (\ref{eq:G-G}) can be bounded by using Lemma \ref{lem:estimates}. We find 
\begin{equation}
\begin{split}
\Big\| \frac{1}{N}\int dx &dy\, V(x-y) a^{*}_{l}(u_{x})a^{*}_{l}(u_{y})a^{*}_{r}(\overline{v_{x}})a_{l}(u_{y}) \widetilde \cU_{N}(t;s)\xi \Big\| \\
&\leq \frac{1}{N}\int dp\, |\hat V(p)| \Big\| \int dx\, a^{*}_{l}(u_{x})a^{*}_{r}(\overline{v_{x}})e^{ipx} \int dy\, e^{-ipy} a^{*}_{l}(u_{y})a_{l}(u_{y}) \widetilde{\cU}_{N}(t;s)\O \Big\| \\
&\leq \frac{1}{N}\int dp\,|\hat V(p)| \| u e^{ipx} v \|_{HS} \| u e^{-ipx} u \|_{op} \| (\cN+1)^{3/2} \widetilde{\cU}_{N}(t;s)\xi \| \\
& \leq \frac{C}{N^{1/2}}\| (\cN+1)^{3/2} \widetilde{\cU}_{N}(t;s) \xi \|.
\end{split}
\end{equation}
It is easy to see that the same estimate holds true for all the other contributions arising from (\ref{eq:G-G}). Therefore, we have:
\be
\| (\mathcal{G}_{N}(t) - \widetilde{\mathcal{G}}_{N}(t)) \widetilde \cU_{N}(t;s)\xi \| \leq \frac{C}{N^{1/2}}\|(\cN+1)^{3/2} \widetilde{\cU}_{N}(t;s) \xi \|. \label{mix08a}
\ee
Substituting  this estimate in (\ref{mix08}) and using Proposition \ref{thm:growth}, we get
\be
\Big| \partial_{t}\| \cU_{N}(t;s)\xi - \widetilde \cU_{N}(t;s)\xi \| \Big| \leq N^{-1/6} \exp(c_{2}\exp(c_{1}|t-s|)) \|(\cN+1)^{3/2} \xi \|.\nn
\ee
The integration yields the claim.
\end{proof}

\subsection{Proof of Theorem \ref{thm:maingrowth}}

We claim that there exists a constant $c_1 > 0$ and, for every $k \in \bN$, a constant $c_2 > 0$ such that 
\begin{equation}\label{eq:ind}
\begin{split}
&\left\langle \cU_N (t;s)  \xi_N, (\cN+1)^k \, \cU_N (t;s) \xi_N \right\rangle \\ &\hspace{3cm} \leq  N^{k-n/3} \exp (2(n+1) c_2 \exp (c_1 |t-s|)) \, \langle \xi_N, (\cN+1)^{k+3n} \xi_N \rangle \,
\end{split}
\end{equation}
for every $n \in \bN$ with $n \leq 3k$. We prove the claim (\ref{eq:ind}) by induction over $n$. For $n = 0$, we observe from (\ref{eq:flu-dyn}) that  
\begin{equation}\label{eq:k0-bd} \left\langle \cU_N (t;s) \xi_N, (\cN+1)^k \cU_N (t;s) \xi_N \right\rangle = \left\langle R_t^* e^{-i\cL_N (t-s)} R_s \xi_N, (\cN+1)^k R_t^* e^{-i\cL_N (t-s)} R_s \xi_N \right\rangle .
\end{equation}
Since 
\[ R_t^* \cN  R_t = \cN + 2N -2d\Gamma (\omega_{N,t}  \oplus \bar{\omega}_{N,t}) +2\int dxdy\, \left[ a^{*}_{x,r}a^{*}_{y,l}\overline{u}_{N,t}\overline{v}_{N,t}(x,y) + \text{h.c.} \right] \]
we easily get
\[ R_t^* (\cN+c) R_t \leq C (\cN + N + c) \]
for a universal constant $C > 0$ and any $c\geq 0$. Also, iterating this bound $k$ times (and always shifting the factors of $\cN$ to the left and to the right), we obtain
\be\label{eq:RNR}
R_{t}^{*} (\cN + c)^{k} R_{t}  \leq C( \cN + N + c)^{k} 
\ee
for a $k$-dependent constant $C > 0$.
Hence, applying (\ref{eq:RNR}) twice (once for $R_{t}$, once for $R_{s}$) and using the fact that $\cN$ commutes with the Liouville operator $\cL_N$, we find from (\ref{eq:k0-bd}) that
\[ \begin{split} \left\langle \cU_N (t;s)\xi_N, (\cN+1)^k \cU_N (t;s) \xi_N \right\rangle &\leq C N^k \langle \xi_N , (\cN + 1)^k 
\xi_N \rangle \\ &\leq N^k \exp (2c_2 \exp (c_1 |t-s|)) \langle \xi_N , (\cN + 1)^k \xi_N \rangle \end{split} \]
if $c_2 > 0$ is large enough, depending on $k$. This proves (\ref{eq:ind}) in the case $n=0$.

Next we assume (\ref{eq:ind}) to hold, and we prove it with $n$ replaced by $(n+1)$, assuming $n+1 \leq 3k$. To this end, we write
\be\label{proof6}
\begin{split}
\langle \xi_{N}, \cU_{N}^{*}(t;s)(\cN + 1)^{k} & \cU_{N}(t;s)\xi_{N} \rangle \\ = \; &\langle \xi_{N}, \widetilde{\cU}_{N}^{*}(t;s)(\cN+1)^{k} \widetilde{\cU}_{N}(t;s)\xi_{N} \rangle\\
& + 2\,\re\langle \xi_{N}, \big(\cU_{N}^{*}(t;s) - \widetilde{\cU}_{N}^{*}(t;s)\big)(\cN+1)^{k} \widetilde{\cU}_{N}(t;s)\xi_{N} \rangle\\
& + \langle \xi_{N}, \big(\cU_{N}^{*}(t;s) - \widetilde{\cU}_{N}^{*}(t;s)\big)(\cN+1)^{k} \big(\cU_{N}(t;s) - \widetilde{\cU}_{N}(t;s)\big)\xi_{N} \rangle \\= \; & \text{I} + \text{II} + \text{III} .
\end{split}
\ee
Using Proposition \ref{thm:growth}, we find
\begin{equation}\label{eq:termI} \begin{split}  \text{I} &\leq \exp (c_2 \exp (c_1 |t-s|)) \langle \xi_N, (\cN+1)^k \xi_N \rangle \\ &\leq  N^{k- (n+1)/3} \exp \left((2n+3) c_2 \exp (c_1 |t-s|) \right)  \langle \xi_N , (\cN +1)^{k+3(n+1)} \xi_N \rangle .\end{split} 
\end{equation}
If $k = 0$, the second term on the r.h.s of (\ref{proof6}) can be bounded trivially. For $k \geq 1$, on the other hand, we note that, for a vector $\xi_{N} = \chi(\cN \leq CN)\xi_{N}$:
\be\label{eq:p6b}
\begin{split}
|\text{II}|  & \leq 2\| (\cU_{N}(t;s) - \widetilde{\cU}_{N}(t;s))\xi_{N} \|\| (\cN + 1)^{k}\widetilde{\cU}_{N}(t;s)\xi_{N} \| \\
&\leq 2N^{-1/6} \exp(c_{2}\exp(c_{1}|t-s|)) \langle \xi_N, (\cN+1)^3 \xi_N \rangle^{1/2} \langle \xi_N , (\cN+1)^{2k} \xi_N \rangle^{1/2} \\
&\leq 2N^{-1/6} \exp(c_{2}\exp(c_{1}|t-s|)) \langle \xi_{N}, (\cN + 1)^{2k+1}\xi_{N} \rangle \\
&\leq 2(CN + 1)^{^{k - (n+1)/3}} \exp(c_{2}\exp(c_{1}|t-s|)) \langle \xi_N, (\cN+1)^{k+(n+1)/3+1} \xi_N \rangle \\
&\leq \widetilde C N^{k-(n+1)/3} \exp ((2n+3) c_2 \exp (c_1 |t-s|)) \langle \xi_N, (\cN+1)^{k+3(n+1)} \xi_N \rangle 
\end{split}
\ee
where $\widetilde C>0$ is a $k$-dependent constant. In (\ref{eq:p6b}) we used: Propositions \ref{thm:growth} and  \ref{prop:aux} in the second inequality; the fact that $\max \{2k,3 \} \leq 2k+1$ for $k \geq 1$ in the third inequality; the assumption $\xi_{N} = \chi(\cN \leq CN)\xi_{N}$ in the fourth inequality; and the fact that $n/3+1 \leq 3n$ for all $n \geq 1$ in the last estimate.

Finally, we consider the last term in (\ref{proof6}). To control it, we use the Duhamel formula
\[
\cU_{N}(t;s) - \widetilde{\cU}_{N}(t;s) = \frac{1}{i\e}\int_{s}^{t}ds' \, \cU_{N}(t;s')(\mathcal{G}_{N}(s') - \widetilde{\mathcal{G}}_{N}(s'))\widetilde{\cU}_{N}(s';s).
\]
Setting $\D \mathcal{G}_{N}(s') = \mathcal{G}_{N}(s') - \widetilde{\mathcal{G}}_{N}(s')$, we rewrite $\text{III}$ as
\be\label{eq:p7}
\frac{1}{\e^{2}}\int_{s}^{t}ds' \int_{s}^{t}ds'' \,\langle \cU_{N}(t;s')\D \mathcal{G}_{N}(s')\widetilde{\cU}_{N}(s';s)\xi_{N}, (\cN + 1)^{k} \cU_{N}(t;s'')\D \mathcal{G}_{N}(s'')\widetilde{\cU}_{N}(s'';s)\xi_{N} \rangle.
\ee
By Cauchy-Schwarz we find 
\be\label{eq:p7c}
\begin{split}
|\text{III}| & \leq \frac{2|t-s|}{\e^{2}}\int_{s}^{t} ds' \,\langle \cU_{N}(t;s')\D \mathcal{G}_{N}(s')\widetilde{\cU}_{N}(s';s) \xi_{N}, (\cN + 1)^{k} \cU_{N}(t;s')\D \mathcal{G}_{N}(s')\widetilde{\cU}_{N}(s';s)\xi_{N} \rangle,
\end{split}
\end{equation}
and the induction assumption yields 
\[ \begin{split} 
|\text{III}| \leq \; &\frac{2|t-s| N^{k-n/3}}{\e^{2}}\int_{s}^{t}ds' \, \exp (2(n+1) c_2 \exp (c_1 |t-s'|)) \\ &\hspace{3cm} \times \langle \D \mathcal{G}_{N}(s')\widetilde{\cU}_{N}(s';s)\xi_{N}, (\cN + 1)^{k+3n} \D \mathcal{G}_{N}(s')\widetilde{\cU}_{N}(s';s)\xi_{N}\rangle.
\end{split} 
\]
Writing $\D \mathcal{G}_N (s') = A(s') + A^* (s')$, with 
\[\begin{split}  A(s') =  \; &\frac{1}{N}\int dxdy\, V(x-y)\Big(a^{*}_{l}(u_{x})a^{*}_{l}(u_{y})a^{*}_{r}(\overline{v_{x}})a_{l}(u_{y}) - a^{*}_{l}(u_{x})a^{*}_{r}(\overline{v_{y}})a^{*}_{r}(\overline{v_{x}})a_{r}(\overline{v_{y}}) \\
& \hspace{3.5cm} + a^{*}_{r}(\overline{u_{x}})a^{*}_{r}(\overline{u_{y}})a^{*}_{l}(v_{x})a_{r}(\overline{u_{y}}) - a_{r}^{*}(\overline{u_{x}})a^{*}_{l}(v_{y})a^{*}_{l}(v_{x})a_{l}(v_{y}) \Big) \end{split} \]
we conclude that, using $A(s')\cN = (\cN-2) A(s')$ and Cauchy-Schwarz inequality:
\[\begin{split} &|\text{III}| \leq \; \frac{C |t-s| N^{k-n/3}\exp (2(n+1) c_2 \exp (c_1 |t-s|))}{\e^{2}} \\ &\quad \times \int_{s}^{t}ds' \,  \left\{  \| A(s') (\cN+1)^{(k+3n)/2} \widetilde{\cU}_N (s';s) \xi_N \|^{2}  + \| A^* (s') (\cN+1)^{(k+3n)/2} \widetilde{\cU}_N (s';s) \xi_N \|^{2}  \right\} \end{split} \]
for a $k$-dependent constant $C>0$. Proceeding as in the proof of Proposition \ref{prop:aux} (see (\ref{mix08a})) and applying Proposition \ref{thm:growth}, we find 
\be\label{eq:p9}
\begin{split}
|\text{III}| 
& \leq C|t-s|^2 N^{k-(n+1)/3} \exp ((2n+3)c_2 \exp (c_1 |t-s|)) 
\| (\cN + 1)^{(k + 3(n+1))/2} \xi_{N}\|^{2} 
\end{split}
\end{equation} 
for an appropriate $k$-dependent constant $C >0$. Combining the last bound with (\ref{eq:termI}) and (\ref{eq:p6b}) we can estimate $|\text{I} + \text{II} + \text{III}|$ as follows:
\be
(1 + \widetilde C +C|t-s|^2) N^{k-(n+1)/3} \exp ((2n+3)c_2 \exp (c_1 |t-s|))\| (\cN + 1)^{(k + 3(n+1))/2} \xi_{N}\|^{2}.
\ee
If $c_2 > 0$ is large enough, we have
\[ (1 + \widetilde C +C |t-s|^2) \leq \exp (c_2 \exp (c_1 |t-s|)) \]
for all $t,\,s \in \bR$. Hence, 
\[ |\text{I} + \text{II} + \text{III}| \leq N^{k-(n+1)/3} \exp (2(n+2)c_2 \exp (c_1 |t-s|))\| (\cN + 1)^{(k + 3(n+1))/2} \xi_{N}\|^{2} \]
which concludes the proof of (\ref{eq:ind}) with $n$ replaced by $(n+1)$. Finally, putting $n=3k$ and $s=0$ in (\ref{eq:ind}), we obtain
\[ \langle \xi_N, \cU_N^* (t;0) (\cN+1)^k \cU_N (t;0) \xi_N \rangle \leq \exp (c_2 \exp (c_1 |t|)) \langle \xi_N, (\cN+1)^{10k} \xi_N \rangle \]
for an appropriate newly defined constant $c_2$ depending on $k$. This completes the proof.\qed


\section{Proof of main results}\label{sec:proof}
\setcounter{equation}{0}

In this section we prove  Theorems \ref{thm:main} and \ref{thm:maink}.

\begin{proof}[Proof of Theorem \ref{thm:main}] We start from the expression
\be\label{proof1}
\begin{split}
\g^{(1)}_{N,t}(x;y) &= \langle \psi_{N,t}, a^{*}_{y,l}a_{x,l} \psi_{N,t} \rangle \\
&= \langle e^{-i\cL_{N}t/\e}R_{0}\xi_{N}, a^{*}_{y,l}a_{x,l} e^{-i\cL_{N}t/\e}R_{0}\xi_{N} \rangle \\
&= \langle \xi_{N}, R_{0}^{*} e^{i\cL_{N}t/\e} a^{*}_{y,l}a_{x,l} e^{-i\cL_{N}t/\e}R_{0}\xi_{N} \rangle.
\end{split}
\ee
Introducting the fluctuation dynamics $\cU_{N}(t;0) = R^{*}_{t}e^{-i\cL_{N}t/\e}R_{0}$, we can rewrite (\ref{proof1}) as
\[
\begin{split}
\g^{(1)}_{N,t}(x;y) &= \langle \xi_{N}, \cU_{N}(t;0)^{*}R^{*}_{t} a^{*}_{y,l} R_{t}R^{*}_{t} a_{x,l} R_{t} \cU_{N}(t;0)\xi_{N} \rangle\\
&= \langle \xi_{N}, \cU_{N}(t;0)^{*}(a^{*}_{l}(u_{y}) - a_{r}(\overline{v_{y}}))(a_{l}(u_{x}) - a^{*}_{r}(\overline{v_{x}})) \cU_{N}(t;0)\xi_{N}\rangle\\
& =  \langle \xi_{N}, \cU_{N}(t;0)^{*}\{ a^{*}_{l}(u_{y})a_{l}(u_{x}) - a^{*}_{l}(u_{y})a^{*}_{r}(\overline{v_{x}}) - a_{r}(\overline{v_{y}})a_{l}(u_{x}) - \\
& \qquad\qquad\qquad\qquad\qquad - a^{*}_{r}(\overline{v_{x}})a_{r}(\overline{v_{y}}) + \langle v_{x}, v_{y}  \rangle \}\cU_{N}(t;0)\xi_{N}\rangle.
\end{split}
\]
Observe that
\be
\langle v_{x}, v_{y}  \rangle = \int dz\, \overline{v_{N,t}(z,x)}v_{N,t}(z,y)  = \o_{N,t}(x;y)\, . 
\ee
Therefore, we can write
\[
\begin{split}
\g^{(1)}_{N,t}(x;y) - \o_{N,t}(x;y) & = \langle \xi_{N}, \cU_{N}(t;0)^{*}\{ a^{*}_{l}(u_{y})a_{l}(u_{x}) - a^{*}_{l}(u_{y})a^{*}_{r}(\overline{v_{x}}) - \\
& \hspace{7em} - a_{r}(\overline{v_{y}})a_{l}(u_{x}) - a^{*}_{r}(\overline{v_{x}})a_{r}(\overline{v_{y}})\}\cU_{N}(t;0)\xi_{N}\rangle\;.
\end{split}
\]
Integrating against the kernel of a one-particle Hilbert-Schmidt observable $O$, we find
\be\label{proof6bb}
\begin{split}
\tr\, O(\g^{(1)}_{N,t} - \o_{N,t}) &= \langle \xi_{N}, \cU_{N}^{*}(t;0)\big(d\G_{l}(u_{t}Ou_{t}) - d\G_{r}(\overline v_{t}\overline{O}^{*}v_{t})\big)\cU_{N}(t;0)\xi_{N} \rangle \\
& - 2\Re\, \langle \xi_{N}, \cU_{N}^{*}(t;0)\int dz_{1}dz_{2}\, a_{r,z_{1}}a_{l,z_{2}}\, (v_{t} O u_{t})(z_{1};z_{2})\,\cU_{N}(t;0)\xi_{N} \rangle .
\end{split}
\ee
Combining the estimates of Lemma \ref{lem:estimates} with the bounds $\| u_{t} \|\leq 1$, $\| v_{t} \|\leq 1$, we conclude that
\be\label{proof7}
\begin{split}
\Big| \tr\, O(\g^{(1)}_{N,t} - \o_{N,t}) \Big| \leq\; & (\| u_{t} O u_{t} \| + \|\overline v_{t}\overline{O}^{*}v_{t}\|) \langle \xi_{N}, \cU_{N}^{*}(t;0)\cN\cU_{N}(t;0)\xi_{N} \rangle + \\
& + 2\| v_{t} O u_{t} \|_{HS} \big\| \cN^{1/2} \cU_{N}(t;0)\xi_{N}\big\|\|\xi_{N}\|\\
\leq & C\| O \|_{HS} \langle \xi_{N}, \cU_{N}^{*}(t;0)\cN\cU_{N}(t;0)\xi_{N} \rangle.
\end{split}
\ee
Finally, from Theorem \ref{thm:maingrowth} we find
\[
\|\g^{(1)}_{N,t} - \o_{N,t}\|_{HS} \leq \exp(c_{2}\exp(c_{1}|t|))\langle \xi_{N}, (\cN + 1)^{10} \xi_{N}\rangle\;,
\]
which concludes the proof of convergence in the Hilbert-Schmidt norm. To prove convergence in trace norm, we proceed as follows. Let $O$ be a compact operator on $\h = L^{2}(\bR^{3})$. Using the bound $\| u_{t} O v_{t} \|_{HS}\leq \| O \|\|v_{t}\|_{HS}\leq N^{1/2}  \| O \|$, one notices that (\ref{proof6bb}) can be alternatively estimated by 
\[
\Big| \tr\, O(\g^{(1)}_{N,t} - \o_{N,t}) \Big| \leq C\|O\|N^{1/2}\langle \xi_{N}, \cU_{N}^{*}(t;0)\cN\cU_{N}(t;0)\xi_{N} \rangle\;,
\]
and the proof is concluded by applying again Theorem \ref{thm:maingrowth}.
\end{proof}

Next, we show the convergence of the $k$-particle reduced density, as stated in Theorem~\ref{thm:maink}.

\begin{proof}[Proof of Theorem \ref{thm:maink}]
We start from the expression
\be\label{eq:k}
\begin{split} \gamma_{N,t}^{(k)} &(x_1, \dots , x_k ; x'_1, \dots x'_k) \\ &= \big\langle e^{-i \cL_N t/ \e} 
R_0 \xi_{N},  a^*_{x'_k,l}  \dots a^*_{x'_1,l} a_{x_1,l} \dots a_{x_k,l} e^{-i\cL_{N} t/\e} R_0 \xi_{N} \big\rangle \\ &=  \big\langle \cU_N (t;0) \xi_{N}, R_t^* a^*_{x'_k,l}  \dots a^*_{x'_1,l} a_{x_1,l} \dots a_{x_k,l} R_t \cU_N (t;0)  \xi_{N} \big\rangle \\
&=  \big\langle \cU_N (t;0) \xi_{N}, \big(a^*_{l}(u_{t,x'_k}) - a_{r}(\cc v_{t,x'_k})\big)\cdots\big(a^*_{l}(u_{t,x'_1}) - a_{r}(\cc v_{t, x'_1})\big) \\
& \qquad\qquad\times \big(a_{l}(u_{t,x_1}) - a_{r}^*(\cc v_{t,x_1})\big)\cdots\big( a_{l}(u_{t,x_k}) - a_{r}^*(\cc v_{t,x_k}) \big) \, \cU_N (t;0) \xi_{N} \big\rangle.
\end{split}
\ee
The strategy of the proof is to expand the  last product  into a sum of $2^{k}$ terms. Each summand can be  put into normal order, using Wick's theorem. The fully contracted term will give rise to the kernel of $\o^{(k)}_{N,t}$, while all the others will be proven to be of smaller order in $N$. The proofs 
of these facts follow the proof of Theorem 2.2 in \cite{BPS}. Notice that, in contrast to \cite{BPS}, for a mixed state the operators $u_{N,t}$ and $\overline{v}_{N,t}$ are {\it not} orthogonal. However, the  arguments of the fermionic operators  $u_{N,t}\oplus 0$ and $0\oplus \overline v_{N,t}$ are orthogonal operators in $\h\oplus \h$, and that suffices for our purposes.  In order to make the proof of Theorem \ref{thm:main} self-contained, we sketch the main steps.

\medskip

Proceeding as in the proof of Theorem 2.2 in \cite{BPS}, we expand the products in Eq.\ (\ref{eq:k}) and we put the result into normal order. The 
end result is a sum of contributions of the following form:
\be\label{eq:typicalwick}
\begin{split}
& \pm \Big\langle \cU_N (t;0) \xi_{N}, \, :a^\sharp_{\a'_{1}}(w_1(\cdot;x'_{\sigma(1)}))\cdots a^\sharp_{\a'_{k-j}}(w_{k-j}(\cdot;x'_{\sigma(k-j)})) \\ &\hspace{2.5cm} \times a^\sharp_{\a_{1}}(\eta_1(\cdot;x_{\pi(1)}))\cdots a^\sharp_{\a_{k-j}}(\eta_{k-j}(\cdot;x_{\pi(k-j)})): \, \cU_N (t;0) \xi_{N} \Big\rangle \\
& \hspace{5cm} \times\omega_{N,t} (x_{\pi(k-j+1)};x'_{\sigma(k-j+1)})\cdots\omega_{N,t} (x_{\pi(k)};x'_{\sigma(k)})
\end{split}
\ee
where the symbol $:\cdots :$ denotes normal ordering; $j \leq k$ denotes the number of contractions; $\pi, \sigma \in S_k$ are two appropriate permutations; for every $s=1, \dots , k-j$, $w_s, \eta_s: \h \to \h$ are either the operator $u_t$ or the operator $\cc v_t$ (the operators are identified with their integral kernels); the $\a$, $\a'$ labels are either $l$ or $r$.

\medskip

{\it Fully contracted term.} For $j=k$, the contractions produce the kernel of the $k$-particle density matrix $\o^{(k)}_{N,t}$:
\[
\sum_{\pi\in S_{k}} \s_{\pi}\o_{N,t}(x_{1};x'_{\pi(1)})\cdots \o_{N,t}(x_{k};x'_{\pi(k)}) = \o^{(k)}_{N,t}(\mathbf{x}_{k};\mathbf{x}'_{k})
\]
where $\mathbf{x}_{k} = (x_{1},\ldots x_{k})$, $\mathbf{x}'_{k} = (x'_{1},\ldots x'_{k})$.

\medskip

{\it Error terms.} Consider now the contributions corresponding to $j<k$. Let $O$ be a Hilbert-Schmidt operator on $L^2 (\bR^{3k})$, with integral kernel $O(\bx_k ; \bx'_k)$. Integrating \eqref{eq:typicalwick} against $O(\bx_k ; \bx'_k)$, we get 
\be\label{eq:wickbound}
\begin{split}
\text{I} & := \Big|  \int d \bx_k d\bx'_k \, O (\bx_k;\bx'_k) \, \big\langle \cU_N (t;0) \xi_{N}, \, :a^\sharp_{\a'_{1}}(w_1(\cdot;x'_{\sigma(1)}))\cdots a^\sharp_{\a'_{k-j}}(w_{k-j}(\cdot;x'_{\sigma(k-j)})) \\
&\hspace{2.5cm} \times a^\sharp_{\a_{1}}(\eta_1(\cdot; x_{\pi(1)}))\cdots a^\sharp_{\a_{k-j}}(\eta_{k-j}(\cdot; x_{\pi(k-j)})): \cU_N (t;0)\xi_{N} \big\rangle \\
& \hspace{5cm} \times\omega_{N,t} (x_{\pi(k-j+1)};x'_{\sigma(k-j+1)})\cdots\omega_{N,t} (x_{\pi(k)};x'_{\sigma(k)})\Big|. 
\end{split} 
\ee
This contribution can be rewritten as
\[ 
\begin{split} 
\text{I} & = \Big| \int d\bx_k d \bx'_k \left[ \eta^{(\pi (1))}_1\cdots \eta^{(\pi (k-j))}_{k-j}  \, O \, w^{(\sigma (1))}_1 \cdots w^{(\sigma (k-j))}_{k-j} \right]  (\bx_k ; \bx'_k) \\
& \hspace{1.0cm} \times \big\langle \cU_N (t;0) \xi_{N}, \, :a^\sharp_{x'_{\sigma(1)},\a'_{1}} \cdots a^\sharp_{x'_{\sigma(k-j)},\a'_{k-j}} a^\sharp_{x_{\pi(1)},\a_{1}} \cdots a^\sharp_{x_{\pi(k-j)},\a_{k-j}}: \, \cU_N (t;0) \xi_{N} \big\rangle\\
&\hspace{3cm} \times\omega_{N,t} (x_{\pi(k-j+1)}; x'_{\sigma(k-j+1)})\cdots \omega_{N,t} (x_{\pi(k)}; x'_{\sigma(k)})\Big| 
\end{split}
\]
where $\eta_\ell^{(\pi (\ell))}$ and $w_\ell^{(\sigma (\ell))}$ denote the one-particle operators $\eta_\ell$ and $w_\ell$ acting only on particle $\pi (\ell)$ and particle $\sigma (\ell)$, respectively.  Notice that some of the operators $\eta_\ell^{(\pi(\ell))}$ and $w_\ell^{(\sigma (\ell))}$ should actually be replaced by their transpose or their complex conjugate. This change however does not affect our analysis, since we will only need the bounds $\| \eta_j \|, \| w_j \| \leq 1$ for the operator norms.  Applying H\"older's inequality, we get
\[
\begin{split} 
&\text{I} \leq \;  \big\|  \eta^{(\pi (1))}_1\cdots \eta^{(\pi (k-j))}_{k-j}  \, O \, w^{(\sigma (1))}_1 \cdots w^{(\sigma (k-j))}_{k-j} \big\|_{\text{HS}}\times \\ &  \bigg( \int d\bx_k d \bx'_k \, \big| \big\langle \cU_N (t;0) \xi_{N}, \, :a^\sharp_{x'_{\sigma(1)},\a'_{1}}\cdots a^\sharp_{x'_{\sigma(k-j),\a'_{k-j}}}a^\sharp_{x_{\pi(1)},\a_{1}}\cdots a^\sharp_{x_{\pi(k-j)},\a_{k-j}}: \cU_N(t;0)\xi_{N}\big\rangle \big|^2 \\ & \hspace{2cm} \times  \big| \omega_{N,t} (x_{\pi(k-j+1)};x'_{\sigma(k-j+1)}) \big|^2 \cdots \big|\omega_{N,t} (x_{\pi(k)},;x'_{\sigma(k)}) \big|^2 \bigg)^{1/2}\\
 \leq \; &  \| O \|_{\text{HS}} \| \omega_{N,t} \|_{\text{HS}}^j  \\
&\times \bigg( \int \di x_{\pi(1)}\cdots \di x_{\pi(k-j)}\di x'_{\sigma(1)}\cdots dx'_{\sigma(k-j)} \\
&\hspace{0.5cm} \times  \big| \big\langle \cU_N (t;0) \xi_{N}, :a^\sharp_{x'_{\sigma(1)},\a'_{1}}\cdots a^\sharp_{x'_{\sigma(k-j),\a'_{k-j}}}a^\sharp_{x_{\pi(1)},\a_{1}}\cdots a^\sharp_{x_{\pi(k-j)},\a_{k-j}}: \cU_N (t;0) \xi_{N} \big\rangle \big|^2 \bigg)^{1/2}.
\end{split}\]
Since $\| \omega_{N,t} \|_{\text{HS}} \leq N^{1/2}$ and since the operators in the inner product are normal ordered, we obtain (possibly by moving some factors of $(\cN+1)$ and $(\cN+1)^{-1}$ around, to re-equilibrate  the number of creation and annihilation operators left):
\[ \text{I} \leq C \| O \|_{\text{HS}} N^{j/2} \langle \cU_N (t;0) \xi_{N} , (\cN+1)^{k-j} \cU_N (t;0) \xi_{N} \rangle \,.\]
Therefore, the contributions of the contractions with $j<k$ arising from (\ref{eq:k}) after applying Wick's theorem and integrating against a Hilbert-Schmidt operator $O$ can be bounded by
\begin{equation}\label{eq:jmink} C \| O \|_{\text{HS}} \, N^{(k-1)/2} \langle \cU_N (t;0) \xi_{N}, (\cN+1)^{k} \cU_N (t;0) \xi_{N} \rangle \,. \end{equation}
This concludes the proof of convergence of many-body quantum dynamics to Hartree-Fock dynamics in Hilbert-Schmidt norm.

\medskip

To prove convergence in trace norm, we proceed as follows. The goal is to prove that, for every compact operator $O$ on $L^{2}(\bR^{3k})$, the contribution (\ref{eq:wickbound}) is estimated as follows:
\begin{equation}\label{eq:term-tr} \text{I} \leq C \| O \| \, N^{\frac{k+j}{2}} \exp (c_1 \exp (c_2 |t|)) \end{equation}
for all $\xi_{N} \in \cF(\h \oplus \h)$ with $\langle \xi_{N}, \cN^{10k} \xi_{N} \rangle < \infty$, and the number of contractions $0 \leq j < k$. In fact, because of the fermionic symmetry of $\gamma_{N,t}^{(k)}$ and $\omega_{N,t}^{(k)}$, it is enough to establish (\ref{eq:term-tr}) for all bounded $O$ with the symmetry
\[ O (x_{\pi (1)}, \dots , x_{\pi (k)} ; x'_{\sigma (1)}, \dots , x'_{\sigma (k)}) = \text{sgn} (\pi) \text{sgn} (\sigma) \, O (x_1, \dots , x_k; x'_1, \dots x'_k) \]
for any permutations $\pi, \sigma \in S_k$.
For such observables, (\ref{eq:wickbound}) can be rewritten as
\[ \begin{split}  \text{I} & = \Big|  \int d \bx_k d\bx'_k \, O (\bx_k,\bx'_k) \, \big\langle \cU_N (t;0) \xi_{N}, \, :a^\sharp_{\a'_{1}}(w_1(\cdot,x'_1))\cdots a^\sharp_{\a'_{k-j}}(w_{k-j} (\cdot,x'_{k-j})) \\
&\hspace{2.5cm} \times a^\sharp_{\a_{1}}(\eta_1(\cdot, x_{1}))\cdots a^\sharp_{\a_{k-j}}(\eta_{k-j}(\cdot, x_{k-j})): \cU_N (t;0)\xi_{N} \big\rangle \\
& \hspace{5cm} \times\omega_{N,t} (x_{k-j+1},x'_{k-j+1})\cdots\omega_{N,t} (x_{k},x'_{k})\Big| \\
& = \Big| \int d\bx_{k-j} d \bx'_{k-j} \left[ \eta^{(1)}_1\cdots \eta^{(k-j)}_{k-j}  \,\left(  \tr_{k-j+1, \dots , k} \, O  (1 \otimes \omega_{N,t}^{\otimes j}) \right) \, w^{(1)}_1 \cdots w^{(k-j)}_{k-j} \right]  (\bx_{k-j} ; \bx'_{k-j}) \\ & \hspace{2.5cm} \times \big\langle \cU_N (t;0) \xi_{N}, \, :a^\sharp_{x'_1,\a'_{1}} \cdots a^\sharp_{x'_{k-j},\a'_{k-j}} a^\sharp_{x_{1},\a_{1}} \cdots a^\sharp_{x_{k-j},\a_{k-j}}: \, \cU_N (t;0) \xi_{N} \big\rangle
\end{split}\]
where 
\[\begin{split}  \Big( \tr_{k-j+1, \dots , k}  O  &(1 \otimes \omega_{N,t}^{\otimes j}) \Big) (\bx_{k-j} ; \bx'_{k-j}) \\ & = \int dx_{k-j+1} dx'_{k-j+1} \dots dx_k dx'_k \, O (\bx_k ; \bx'_k) \prod_{\ell=k-j+1}^k \omega_{N,t} (x_\ell; x'_\ell) \end{split} \]
denotes the partial trace over the last $j$ particles. Using Cauchy-Schwarz, we obtain
\bea
\text{I} &\leq& \left\| \eta^{(1)}_{1} \dots \eta^{(k-j)}_{k-j} \, \left(\tr_{k-j+1, \dots , k} \, O \, (1 \otimes \o_{N,t}^{\otimes j}) \right) w^{(1)}_{1} \dots w^{(k-j)}_{k-j} \right\|_{\text{HS}} \, \left\| \cN^{\frac{k-j}{2}} \cU_N(t;0)\xi_{N} \right\|^{2}\nn\\
&\leq&  \left\| \eta^{(1)}_{1} \dots \eta^{(k-j)}_{k-j} \right\|_{\text{HS}} \, \left\| \tr_{k-j+1, \dots , k}  \, O (1 \otimes \o_{N,t}^{\otimes j}) \right\| \, \left\| \cN^{\frac{k-j}{2}}\cU_N(t;0)\xi_{N} \right\|^{2}\nn\\
&\leq& N^{\frac{k-j}{2}} \, \left\| \tr_{k-j+1, \dots , k} \, O (1 \otimes \o_{N,t}^{\otimes j}) \right\| \, \left\| \cN^{\frac{k-j}{2}}\cU_N(t;0)\xi_{N} \right\|^{2}\;\nn
\eea
where in the second line we used that $\|w^{(m)}_{m}\|=1$ for all $m=1, \dots , k-j$. Since
\bea
\left\| \tr_{k-j+1, \dots , k} \, O (1\otimes \o_{N,t}^{\otimes j}) \right\| &=& \sup_{{\substack{\phi, \ph \in L^{2}(\mathbb{R}^{3(k-j)}) \\ \|\phi\| = \| \psi \| \leq 1}}} \left| \left\langle \phi, \left(\tr_{k-j+1, \dots , k} \, O (1 \otimes \o_{N,t}^{\otimes j}) \right) \varphi\right\rangle \right| \nn\\
&=& \sup_{{\substack{\phi, \ph \in L^{2}(\mathbb{R}^{3(k-j)}) \\ \|\phi\| = \| \psi \| \leq 1}}} 
\left| \tr \, O \left(|\ph \rangle \langle \phi| \otimes \o^{\otimes j}_{N,t}\right) \right| \nn\\
&\leq& \left(\tr \lvert \omega_{N,t} \rvert\right)^j \norm{O} \leq N^j \, \|O\| \;,\label{eq:trk2}
\eea
we get
\be
\text{I} \leq N^{\frac{k+j}{2}} \, \|O\| \, \left\| \cN^{\frac{k-j}{2}}\cU_N(t;0)\xi \right\|^{2}\,,\nn
\ee
which, by Theorem \ref{thm:maingrowth}, proves (\ref{eq:term-tr}). This concludes the proof of Theorem  \ref{thm:maink}.
\end{proof}

\appendix

\section{Existence of the auxiliary dynamics}\label{app:derivative}

In this appendix we prove that the auxiliary dynamics in the interaction picture $\widehat{\cU}_{N}(t;s)$ exists, and that is solves the differential equation (\ref{eq:Uhat}). We will proceed as follows.

Let $M\in \mathbb{N}$, $M<\infty$. Consider the regularized generator 
\[
\widehat{\mathcal{G}}^{(M)}_{N}(t) = \chi(\cN \leq M)\widehat{\mathcal{G}}_{N}(t)\chi(\cN \leq M)\;.
\]
Using the estimates of Lemma \ref{lem:estimates}, it is easy to see that $\widehat{\mathcal{G}}^{(M)}_{N}(t)$ is a bounded self-adjoint operator (to see this, it is important to notice that in the definition (\ref{eq:ghat}) of $\widehat{\mathcal{G}}_{N}(t)$ the kinetic terms {\it cancel}; this is the reason why we switched to the interaction picture). Assuming the map $t\mapsto \widehat{\mathcal{G}}_{N}^{(M)}(t)\psi$ to be continuous for all $\psi\in D(\cN^{2})$ (which will be proven below), $\widehat{\mathcal{G}}^{(M)}_{N}(t)$ defines a unitary dynamics through the following differential equation:
\be
i\e\frac{d}{dt} \widehat{\cU}^{(M)}_{N}(t;s) = \widehat{\mathcal{G}}^{(M)}_{N}(t) \widehat{\cU}^{(M)}_{N}(t;s)\;.\label{eq:der1}
\ee
We are interested in the limiting evolution as $M\to \infty$. First, we prove that, for all $\psi \in D(\cN^{3})$ and for fixed $t,\,s$, the sequence $\big\{ \widehat{\cU}^{(M)}_{N}(t;s)\psi \big\}_{M\in \mathbb{N}}$ is Cauchy. Therefore, the limit $\widehat{\cU}_{N}(t;s)\psi = \lim_{M\to\infty}\widehat{\cU}^{(M)}_{N}(t;s)\psi$ exists, and it defines a bounded isometric operator in $D(\cN^{3})$. Also, $\widehat{\cU}_{N}(t;s)$ can be extended to a unitary operator over all $\cF$. Then, assuming $t\mapsto \widehat{\mathcal{G}}^{(M)}_{N}(t)\psi$ to be continuous uniformly in $M$ for $\psi\in D(\cN^{4})$, we prove that:
\be\label{eq:der1b}
\lim_{\delta\to 0}\lim_{M\to\infty}\Big\| \frac{i\e}{\delta}\Big( \widehat{\mathcal{U}}_{N}^{(M)}(t+\delta;s) - \widehat{\mathcal{U}}_{N}^{(M)}(t;s) \Big)\psi - \widehat{\mathcal{G}}_{N}(t)\widehat{\mathcal{U}}_{N}^{(M)}(t;s)\psi \Big\| = 0\qquad \forall \psi\in D(\cN^{4})\;.
\ee
This shows that $\widehat{\cU}_{N}(t;s)\psi = \lim_{M\to\infty}\widehat{\cU}^{(M)}_{N}(t;s)\psi$ is differentiable in $t$ for $\psi\in D(\cN^{4})$, and that Eq. (\ref{eq:Uhat}) holds. Finally, we conclude by proving the continuity assumption on $t\mapsto \widehat{\mathcal{G}}^{(M)}_{N}(t)\psi$.

\medskip

{\it Convergence of $\widehat{\cU}_{N}^{(M)}$ to a unitary dynamics.} Let $\psi\in D(\cN^{4})$. Consider
\be\label{eq:der2}
\Big\| \widehat{\cU}^{(M)}_{N}(t;s)\psi - \widehat{\cU} ^{(M')}_{N}(t;s)\psi\Big\| = \Big\| \psi - \widehat{\cU}^{(M)*}_{N}(t;s)\widehat{\cU}^{(M')}_{N}(t;s)\psi \Big\|\;.
\ee
where we used the unitarity of $\widehat{\cU}_{N}^{(M)}(t;s)$. Without loss of generality, we assume that $M\leq M'$. Using (\ref{eq:der1}), we rewrite (\ref{eq:der2}) as
\[
\begin{split}
\Big\| \frac{1}{i\e}\int_{s}^{t}ds'\, i\e\frac{d}{ds'} &\widehat{\cU}^{(M)*}_{N}(s';s)\widehat{\cU}^{(M')}_{N}(s';s) \psi\Big\| \\
& = \frac{1}{\e}\Big\| \int_{s}^{t}ds'\, \widehat{\cU}^{(M)*}_{N}(s';s)\Big( \widehat{\mathcal{G}}^{(M)}_{N}(s') - \widehat{\mathcal{G}}^{(M')}_{N}(s') \Big)\widehat{\cU}^{(M')}_{N}(s';s)\psi \Big\|\;.
\end{split}
\]
Writing $\chi(\cN \leq M') = \chi(\cN \leq M) + \chi(M<\cN \leq M')$, we see that $\widehat{\mathcal{G}}^{(M)}_{N}(s') - \widehat{\mathcal{G}}^{(M')}_{N}(s')$ is given by a sum of contributions containing at least one $\chi(M<\cN \leq M')$. Using the bounds of Lemma \ref{lem:estimates}, estimating $\chi(M<\cN \leq M') \leq \chi(M<\cN)\leq (\cN / M)$ and proceeding as in the proof of Proposition \ref{thm:growth}, it is not difficult to see that
\be\label{eq:der4}
\begin{split}
\frac{1}{\e}\Big\| \int_{s}^{t}ds'\, \widehat{\cU}^{(M)*}_{N}(s';s)\Big( \widehat{\mathcal{G}}^{(M)}_{N}(s') - &\widehat{\mathcal{G}}^{(M')}_{N}(s') \Big)\widehat{\cU}^{(M')}_{N}(s';s) \psi\Big\|  \\ &\leq C\int_{s}^{t}ds'\, \Big\| (\cN + 1)^{2}\chi(\cN > M) \widehat{\cU}^{(M')}_{N}(s';s)\psi \Big\| \\ &\leq C\int_{s}^{t}ds'\, \frac{1}{M}\Big\| (\cN + 1)^{3}\,\widehat{\cU}^{(M')}_{N}(s';s)\psi \Big\| \\
&\leq \widetilde C|t-s|\frac{1}{M}\big\| (\cN + 1)^{3}\psi \big\|
\end{split}
\ee
for a suitable constant $\widetilde C<\infty$ independent of $M$ (but possibly depending on $N$ and $t$). The last inequality follows from a Gronwall-type estimate, whose proof is completely analogous to the one of Proposition \ref{thm:growth}. We omit the details. Eq. (\ref{eq:der4}) implies that the sequence $\big\{ \widehat{\cU}^{(M)}_{N}(t;s)\psi \big\}_{M\in \mathbb{N}}$ is Cauchy for all $\psi \in D(\cN^{3})\subset D(\cN^{4})$; therefore:
\be
\lim_{M\to\infty} \widehat{\cU}^{(M)}_{N}(t;s)\psi =: \widehat{\cU}_{N}(t;s)\psi\;
\ee
exists. Notice that $\widehat{\cU}_N (t;s)$ defines an isometry on $D(\cN^4)$, since
\[ \| \widehat{\cU}_N (t;s) \psi \| = \lim_{M \to \infty} \| \widehat{\cU}^{(M)}_N (t;s) \psi \| = \| \psi \| \]
As a consequence, $\widehat{\cU}_N (t;s)$ can be extended to an isometry over the full Fock space $\cF (\h \oplus \h)$. For  $\psi \in D(\cN^4)$, we have
\[ \widehat{\cU}^*_N (t;s) \widehat{\cU}_N (t;s) \psi = \lim_{M \to \infty} \widehat{\cU}^*_N (t;s) \widehat{\cU}^{(M)}_N (t;s) \psi \]
{F}rom the invariance of $D(\cN^4)$ with respect to the evolution $\widehat{\cU}^{(M)}_N (t;s)$, we find
\[ \widehat{\cU}^*_N (t;s) \widehat{\cU}_N (t;s) \psi = \lim_{M',M \to \infty} \widehat{\cU}^{(M')*}_N (t;s) \widehat{\cU}^{(M)}_N (t;s) \psi = \psi\]
The last identity follows from (\ref{eq:der4}). This proves that $\widehat{\cU}_{N}(t;s)$ defines a unitary operator over $\cF (\h \oplus \h)$, for all $t,s \in \bR$. 

Finally, we prove (\ref{eq:der1b}). We have
\be\label{eq:der4bb}
\begin{split}
&\Big\| \frac{i\e}{\delta}\Big( \widehat{\mathcal{U}}_{N}^{(M)}(t+\delta;s) - \widehat{\mathcal{U}}_{N}^{(M)}(t;s) \Big)\psi - \widehat{\mathcal{G}}_{N}(t)\widehat{\mathcal{U}}_{N}^{(M)}(t;s)\psi \Big\| \\
&= \Big\| \frac{1}{\delta}\int_{t}^{t+\delta}ds'\, \Big( i\e\frac{d}{ds'}\widehat{\cU}^{(M)}_{N}(s';s) - \widehat{\mathcal{G}}_{N}(t)\widehat{\cU}^{(M)}_{N}(t;s) \Big)\psi \Big\| \\
&= \Big\| \frac{1}{\delta}\int_{t}^{t+\delta}ds'\,\Big( \widehat{\mathcal{G}}^{(M)}_{N}(s')\widehat{\cU}^{(M)}_{N}(s';s)\widehat{\cU}^{(M)*}_{N}(t;s) - \widehat{\mathcal{G}}_{N}(t) \Big) \widehat{\cU}^{(M)}_{N}(t;s)\psi \Big\|\\
&\leq \Big\| \frac{1}{\delta}\int_{t}^{t+\delta}ds'\, \Big( \widehat{\mathcal{G}}^{(M)}_{N}(s') - \widehat{\mathcal{G}}_{N}(t) \Big) \widehat{\cU}^{(M)}_{N}(t;s)\psi\Big\|\\
&\quad + \Big\|\frac{1}{i\delta\eps} \int_{t}^{t+\delta} ds'\, \widehat{\mathcal{G}}^{(M)}_{N}(s')\widehat{\cU}^{(M)}_{N}(s';s)\int_{s'}^{t}ds''\, \widehat{\cU}^{(M)*}_{N}(s'';s) \widehat{\mathcal{G}}^{(M)}_{N}(s'') \widehat{\cU}^{(M)}_{N}(t;s) \psi\Big\| \\
&\leq \frac{1}{\delta}\int_{t}^{t+\delta}ds'\,\Big\|\Big( \widehat{\mathcal{G}}^{(M)}_{N}(s') - \widehat{\mathcal{G}}_{N}(t) \Big) \widehat{\cU}^{(M)}_{N}(t;s)\psi \Big\| + C\delta \big\| (\cN+1)^{4} \psi \big\|
\end{split}
\ee
for a suitable constant $C$, dependent on $t$ and $N$. To get the fourth line of (\ref{eq:der4bb}) we used that
\[
\begin{split}
\widehat{\cU}^{(M)}_{N}(s';s)\widehat{\cU}^{(M)*}_{N}(t;s) &= 1 + \frac{1}{i\e}\int_{s'}^{t}ds''\, \widehat{\cU}^{(M)}_{N}(s';s)i\e\partial_{s''}\widehat{\cU}^{(M)*}_{N}(s'';s) \\
&= 1 + \frac{1}{i\e}\int_{s'}^{t}ds''\, \widehat{\cU}^{(M)}_{N}(s';s)\widehat{\cU}^{(M)*}_{N}(s'';s)\widehat{\mathcal{G}}^{(M)}_{N}(s'')
\end{split}
\]
and we applied Cauchy-Schwarz inequality. The last estimate in (\ref{eq:der4bb}) follows from a Gronwall-type argument, completely analogous to those used so far; we omit the details. Now, notice that
\be\label{eq:der4b}
\begin{split}
\Big\|\Big( \widehat{\mathcal{G}}^{(M)}_{N}(s') - &\widehat{\mathcal{G}}_{N}(t) \Big)\widehat{\cU}^{(M)}_{N}(t;s)\psi \Big\| \\ &\leq \Big\| \Big( \widehat{\mathcal{G}}^{(M)}_{N}(s') - \widehat{\mathcal{G}}^{(M)}_{N}(t) \Big)\widehat{\cU}^{(M)}_{N}(t;s)\psi \Big\| + \frac{C}{M}\big\| (\cN + 1)^{3} \widehat{\cU}^{(M)}_{N}(t;s) \psi \big\| 
\\ &\leq \Big\| \Big( \widehat{\mathcal{G}}^{(M)}_{N}(s') - \widehat{\mathcal{G}}^{(M)}_{N}(t) \Big)\widehat{\cU}^{(M)}_{N}(t;s)\psi \Big\| + \frac{\widetilde{C}}{M}\big\| (\cN + 1)^{3} \psi \big\|
\end{split}
\ee
where we used that $1 = \chi(\cN \leq M) + \chi(\cN > M)$ and we estimated $\chi(\cN > M)\leq (\cN / M)$. We are left with estimating the first term on the r.h.s. of (\ref{eq:der4b}). We write:
\be
\begin{split}
\Big\| \Big( \widehat{\mathcal{G}}^{(M)}_{N}(s') - &\widehat{\mathcal{G}}^{(M)}_{N}(t) \Big) \widehat{\cU}^{(M)}_{N}(t;s)\psi \Big\| \\
&= \Big\| \Big( \widehat{\mathcal{G}}^{(M)}_{N}(s') - \widehat{\mathcal{G}}^{(M)}_{N}(t) \Big)(\cN + 1)^{-2}(\cN + 1)^{2}\widehat{\cU}^{(M)}_{N}(t;s)\psi \Big\| \\
&= \Big\| \Big( \widehat{\mathcal{G}}^{(M)}_{N}(s') - \widehat{\mathcal{G}}^{(M)}_{N}(t) \Big)(\cN + 1)^{-2}\varphi \Big\| 
\end{split}
\ee
where $\varphi := (\cN + 1)^{2}\widehat{\cU}^{(M)}_{N}(t;s)\psi$. A Gronwall-type argument, similar to the one used in the proof of Proposition \ref{thm:growth}, shows that $\varphi \in D(\cN^{2})$. Consider an approximating sequence $\{\varphi_{n}\}$, $\varphi_{n}\in D(\mathcal{K})$, with $\mathcal{K} = d\Gamma_{l}(-\e^{2}\Delta) + d\Gamma_{r}(-\e^{2}\Delta)$, such that $\varphi_{n} \to \varphi$. We have:
\be
\begin{split}
\Big\| \Big( \widehat{\mathcal{G}}^{(M)}_{N}(s') - \widehat{\mathcal{G}}^{(M)}_{N}(t) \Big)(\cN + 1)^{-2}\varphi \Big\| &\leq \Big\| \Big( \widehat{\mathcal{G}}^{(M)}_{N}(s') - \widehat{\mathcal{G}}^{(M)}_{N}(t) \Big)(\cN + 1)^{-2}\varphi_{n} \Big\|\\
&\quad + \Big\| \Big( \widehat{\mathcal{G}}^{(M)}_{N}(s') - \widehat{\mathcal{G}}^{(M)}_{N}(t) \Big)(\cN + 1)^{-2}(\varphi - \varphi_{n}) \Big\| \\
& \leq \Big\| \Big( \widehat{\mathcal{G}}^{(M)}_{N}(s') - \widehat{\mathcal{G}}^{(M)}_{N}(t) \Big)(\cN + 1)^{-2}\varphi_{n} \Big\|\\
&\quad + C\| \varphi - \varphi_{n} \|
\end{split}  
\ee
where in the last inequality we used that $\Big( \widehat{\mathcal{G}}^{(M)}_{N}(s') - \widehat{\mathcal{G}}^{(M)}_{N}(t) \Big)(\cN + 1)^{-2}$ is a bounded operator. Fix $\eta > 0$ and choose $n = n(\eta)$ such that $C\| \varphi - \varphi_{n} \| \leq \eta/2$. Notice that $(\cN + 1)^{-2}\varphi \in D(\cN(\mathcal{K} + \cN))$; as we will prove later, $t\mapsto \widehat{\mathcal{G}}^{(M)}_{N}(t)\psi$ is {\it differentiable} for $\psi\in D(\cN(\mathcal{K} + \cN))$, uniformly in $M$. Thus, we get:
\be\label{eq:der7}
\begin{split}
\Big\| \Big( \widehat{\mathcal{G}}^{(M)}_{N}(s') - \widehat{\mathcal{G}}^{(M)}_{N}(t) \Big)(\cN + 1)^{-2}\varphi_{n} \Big\|\ &\leq C |t-s'|\big\| (\cN + 1)(\mathcal{K} + \cN)(\cN + 1)^{-2}\varphi_{n} \big\| \\
&\leq C \delta\big\| (\mathcal{K} + 1)\varphi_{n} \big\|
\end{split}
\ee
for a constant $C> 0$. Choosing $\delta$ small enough, this term can be estimated by $\eta/2$. Therefore, we have:
\be\label{eq:der6}
\Big\| \Big( \widehat{\mathcal{G}}^{(M)}_{N}(s') - \widehat{\mathcal{G}}^{(M)}_{N}(t) \Big)(\cN + 1)^{-2}\varphi \Big\| \leq \eta\;.
\ee
From (\ref{eq:der4bb}), (\ref{eq:der4b}) and (\ref{eq:der6}), we end up with:
\be
\begin{split}
\Big\| \frac{i\e}{\delta}\Big( \widehat{\mathcal{U}}_{N}^{(M)}(t+\delta;s) - &\widehat{\mathcal{U}}_{N}^{(M)}(t;s) \Big)\psi - \widehat{\mathcal{G}}_{N}(t)\widehat{\mathcal{U}}_{N}^{(M)}(t;s)\psi \Big\| \\
&\leq \eta + C\Big(\frac{1}{M} + \delta\Big)\big\| (\cN + 1)^{4} \widehat{\cU}^{(M)}_{N}(t;s) \psi \big\| \\
&\leq \eta + C\Big(\frac{1}{M} + \delta\Big)\big\| (\cN + 1)^{4} \psi \big\|\;,
\end{split}
\ee
for $\psi\in D(\cN^{4})$ and for all $\delta$ small enough. Letting $\delta \to 0$ and $M \to \infty$, we conclude that 
\[ \lim_{\delta \to \infty} \lim_{M \to 0} \Big\| \frac{i\e}{\delta}\Big( \widehat{\mathcal{U}}_{N}^{(M)}(t+\delta;s) - \widehat{\mathcal{U}}_{N}^{(M)}(t;s) \Big)\psi - \widehat{\mathcal{G}}_{N}(t)\widehat{\mathcal{U}}_{N}^{(M)}(t;s)\psi \Big\| \leq \eta \]
Since $\eta > 0$ is arbitrary, we find (\ref{eq:der1b}).  

\medskip

{\it Proof of (\ref{eq:der7})}. Here we show that the map $t\mapsto \widehat{\mathcal{G}}_{N}^{(M)}(t)\psi$ it is differentiable for all $\psi\in D(\cN(\mathcal{K} + \cN))$. For simplicity, we set $s=0$. Here, the cutoff on the number of particles does not play any role; therefore, we prove the differentiability directly for $t\mapsto \widehat{\mathcal{G}}_{N}(t)\psi$. We have:
\[
\begin{split}
\widehat{\mathcal{G}}_{N}(t) &= \cU_{N}^{(0)*}(t)\Big( d\Gamma_{l}(\rho_{t}*V - X_{t}) - d\Gamma_{r}(\rho_{t}*V - \overline{X}_{t}) \Big)\cU_{N}^{(0)}(t) \\
&\quad + \cU_{N}^{(0)*}(t)\widetilde{\mathcal{Q}}_{N} \cU_{N}^{(0)}(t) + \cU_{N}^{(0)*}(t)\mathcal{C}_{N}\cU_{N}^{(0)}(t) \\
&= \text{I} + \text{II} + \text{III}\;.
\end{split}
\]
We will take the time derivative of these three contributions separately, and we will estimate the action of the resulting operators on $\psi\in D(\cN(\mathcal{K} + \cN))$. In what follows, we shall denote by $C$ a generic positive (finite) constant, possibly depending on $N$ and $t$.

\medskip

Consider first the left contribution to $\text{I}$. Let us focus first on the term containing $\rho_{t}*V$. We have:
\[
\begin{split}
\cU_{N}^{(0)*}(t)d\Gamma_{l}(\rho_{t}*V) \cU_{N}^{(0)}(t) &= \int dp\, \hat V(p)\hat \rho_{t}(p)\, \cU_{N}^{(0)*}(t)d\Gamma_{l}(e^{ip\cdot x})\cU_{N}^{(0)}(t)\\
&= \int dp\, \hat V(p)\hat \rho_{t}(p)\, d\Gamma_{l}\big(e^{ip\cdot (x - i\e\nabla t)}\big)\;.
\end{split}
\]
Thus,
\be\label{eq:ddt_I}
\begin{split}
i\e\frac{d}{dt}\cU_{N}^{(0)*}(t)d\Gamma_{l}(\rho_{t}*V) \cU_{N}^{(0)}(t) &= \int dp\, \hat V(p)\Big(i\e\frac{d}{dt}\hat\rho_{t}(p)\Big)\, d\Gamma_{l}\big(e^{ip\cdot (x - i\e\nabla t)}\big)\\
&\quad + \int dp\, \hat V(p) \hat\rho_{t}(p)\, d\Gamma_{l}\Big(i\e\frac{d}{dt}e^{ip\cdot (x - i\e\nabla t)}\Big)\;.
\end{split}
\ee
To bound the Fourier transform of the density, notice that:
\be\label{eq:ddt_I2}
\begin{split}
i\e\frac{d}{dt}\hat\rho_{t}(p) &= i\e\frac{d}{dt}\tr(e^{ip\cdot x}\omega_{N,t}) \\
&= \tr\Big(e^{ip\cdot x}[-\e^{2}\Delta + \rho_{t}*V - X_{t}, \omega_{N,t}]\Big) \\
&=\tr\Big( (i\e\nabla + \e p)^{2}e^{ip\cdot x}\omega_{N,t} - e^{ip\cdot x}\omega_{N,t}(i\e\nabla)^{2} \Big) + \tr\Big( e^{ip\cdot x}[\rho_{t}*V - X_{t}, \omega_{N,t}] \Big) \\
& = \tr\Big( 2i\e^{2}p\cdot \nabla e^{ipx}\omega_{N,t} + \e^{2}|p|^{2} e^{ip\cdot x}\omega_{N,t} \Big) + \tr\Big( e^{ip\cdot x}[\rho_{t}*V - X_{t}, \omega_{N,t}] \Big)\;. 
\end{split}
\ee
The last term in (\ref{eq:ddt_I2}) is bounded proportionally to $\| \o \|_{\text{tr}}$. Concerning the first term in (\ref{eq:ddt_I2}), it is easy to see that:
\[
\Big| \tr\Big( 2i\e^{2}p\cdot \nabla e^{ipx}\omega_{N,t} + \e^{2}|p|^{2} e^{ip\cdot x}\omega_{N,t} \Big) \Big| \leq C(|p|^{2}+1)\tr((1-\Delta) \omega_{N,t}) < \infty\;.
\]
The last inequality follows from the finiteness of the energy of $\omega_{N,t}$. Regarding the second term in (\ref{eq:ddt_I}), we get
\[
\Big\| \int dp\, \hat V(p) \hat\rho_{t}(p)\, d\Gamma_{l}\Big(i\e\frac{d}{dt}e^{ip\cdot (x - i\e\nabla t)}\Big) \psi \Big\| \leq C(|p|^{2} + 1)\| d\Gamma_{l}(-\Delta + 1)\psi \|\;.
\]
Therefore,
\[
\Big\| i\e\frac{d}{dt} \cU_{N}^{(0)*}(t)d\Gamma_{l}(\rho_{t}*V) \cU_{N}^{(0)}(t) \psi \Big\| \leq C\| (\mathcal{K} + \mathcal{N}) \psi \|\;.
\]
Consider now the contribution to the left part of $\text{I}$ containing $X_{t}$. We have:
\[
\begin{split}
\cU_{N}^{(0)*}(t) d\Gamma_{l}(X_{t}) \cU_{N}^{(0)}(t) &= \frac{1}{N}\int dp\, \hat{V}(p)\, \cU_{N}^{(0)*}(t) d\Gamma_{l}\big( e^{ip\cdot x}\omega_{N,t}e^{-ip\cdot x} \big)\cU_{N}^{(0)}(t) \\
& = \frac{1}{N}\int dp\, \hat{V}(p)\, d\Gamma_{l}\big( e^{ip\cdot (x - i\e\nabla t)}\widehat{\omega}_{N,t}e^{-ip\cdot (x - i\e\nabla t)} \big)\cU_{N}^{(0)}(t)
\end{split}
\]
where $\widehat{\omega}_{N,t}$ is the time-evolution of $\omega_{N}$ in the interaction picture, namely:
\[
i\e\partial_{t}\widehat{\omega}_{N,t} = e^{i\Delta t/\eps}[\rho_{t}*V - X_{t}, \omega_{N,t}]e^{-i\Delta t/\eps}\;.
\]
Using the estimates of Lemma \ref{lem:estimates}, it is not difficult to see that
\[
\Big\| i\e\frac{d}{dt}\cU_{N}^{(0)*}(t) d\Gamma_{l}(X_{t}) \cU_{N}^{(0)}(t) \psi\Big\| \leq C\| \cN^{1/2} \psi \|\;.
\]
The right contribution to $\text{I}$ is estimated in the same way. Thus,
\be\label{eq:estGI}
\Big\| i\e\frac{d}{dt} \cU_{N}^{(0)*}(t)\Big( d\Gamma_{l}(\rho_{t}*V - X_{t}) - d\Gamma_{r}(\rho_{t}*V - \overline{X}_{t}) \Big)\cU_{N}^{(0)}(t) \psi \Big\| \leq C\| (\mathcal{K} + \mathcal{N})\psi \|\;.
\ee
Consider now $\text{II}$. Let us focus on the first term in (\ref{eq:qtilde}). It can be rewritten as follows:
\be\label{eq:ddt_Q}
\begin{split}
&\cU_{N}^{(0)*}(t)\frac{1}{2N}\int dxdy V(x-y) a^{*}_{l}(u_{x})a^{*}_{l}(u_{y})a^{*}_{r}(\overline{v_{y}})a^{*}_{r}(\overline{v_{x}})\,\cU_{N}^{(0)}(t) \\
&= \frac{1}{2N} \int dp\,\hat V(p) \int dz_{1}dz_{2}\, a^{*}_{z_{1},l}a^{*}_{z_{2},r} \Big( \widehat u_{t}e^{ip\cdot (x - i\e\nabla t)} \widehat v_{t} \Big)(z_{1},z_{2}) \\& \quad \times \int dz_{3} dz_{4}\, a^{*}_{z_{3},l}a^{*}_{z_{4},r}\Big( \widehat u_{t} e^{-ip\cdot (x - i\e\nabla t)}\widehat v_{t} \Big)(z_{3},z_{4})
\end{split}
\ee
where $\widehat u_{N,t} = \sqrt{1 - \widehat{\omega}_{N,t}}$ and $\widehat v_{N,t} = \sqrt{\widehat{\omega}_{N,t}}$. Using the estimate (\ref{est2}) in Lemma \ref{lem:estimates}, we get:
\be\label{eq:ddt_Q2}
\Big\| \int dz_{1}dz_{2}\, a^{*}_{z_{1},l}a^{*}_{z_{2},r} \Big( \widehat u_{t}e^{ip\cdot (x - i\e\nabla t)} \widehat v_{t} \Big)(z_{1},z_{2}) \psi \Big\| \leq C\| (\cN + 1)^{1/2}\psi \|\;.
\ee
In order to bound the time derivative of (\ref{eq:ddt_Q}), it is important to notice that
\be\label{eq:ddt_Q2b}
\begin{split}
&e^{-i\Delta t/\eps}\Big(i\e\frac{d}{dt} \widehat u_{t}e^{ip\cdot (x - i\e\nabla t)} \widehat v_{t}\Big)e^{i\Delta t/\eps}\\
&= \Big([\rho_{t}*V - X_{t}, u_{t}] e^{ip\cdot x} v_{t} + u_{t} \Big(i\e\frac{d}{dt}e^{ip\cdot (x - i\e\nabla t)}\Big) v_{t} + u_{t}e^{ip\cdot (x - i\e\nabla t)} [\rho_{t}*V - X_{t},v_{t}] \Big)\;.
\end{split}
\ee
The Hilbert-Schmidt norm of (\ref{eq:ddt_Q2b}) is estimated proportionally to $\tr (1-\Delta)\omega_{N,t}$. Therefore, we find:
\be\label{eq:ddt_Q3}
\Big\| i\e\frac{d}{dt} \int dz_{1}dz_{2}\, a^{*}_{z_{1},l}a^{*}_{z_{2},r} \Big( \widehat u_{t}e^{ip\cdot (x - i\e\nabla t)} \widehat v_{t} \Big)(z_{1},z_{2}) \psi \Big\| \leq \Big(C\tr(1 - \D)\o_{N,t}\Big)\| (\cN + 1)^{1/2} \psi\|\;.
\ee
Using the bounds (\ref{eq:ddt_Q2}), (\ref{eq:ddt_Q3}) we finally get:
\[
\Big\| i\e\frac{d}{dt} \cU_{N}^{(0)*}(t)\frac{1}{2N}\int dxdy V(x-y) a^{*}_{l}(u_{x})a^{*}_{l}(u_{y})a^{*}_{r}(\overline{v_{y}})a^{*}_{r}(\overline{v_{x}})\,\cU_{N}^{(0)}(t) \psi \Big\| \leq C\| (\cN + 1)\psi \|\;.
\]
The other terms in $\text{II}$ can be bounded in exactly the same way. Hence,
\be\label{eq:estGII}
\Big\|  i\e\frac{d}{dt} \cU_{N}^{(0)*}(t)\widetilde{\mathcal{Q}}_{N} \cU_{N}^{(0)}(t) \psi \Big\| \leq C\| (\cN + 1)\psi \|\;.
\ee
We are left with $\text{III}$. Consider the contribution due to the first term in the expression (\ref{eq:CN}) for $\mathcal{C}_{N}$, that is:
\[
\begin{split}
&\cU_{N}^{(0)*}(t) \frac{1}{2N}\int dxdy\, V(x-y)a^{*}_{l}(u_{x})a^{*}_{l}(u_{y})a_{l}(u_{y})a_{l}(u_{x}) \cU_{N}^{(0)}(t) \\
& = \frac{1}{2N}\int dp\, \hat V(p)d\Gamma_{l}\big(\widehat{u}_{t}e^{ip\cdot (x - i\e\nabla t)}\widehat{u}_{t}\big)d\Gamma_{l}\big(\widehat{u}_{t}e^{-ip\cdot (x - i\e\nabla t)}\widehat{u}_{t}\big)\;.
\end{split}
\]
The estimates of Lemma \ref{lem:estimates} imply that
\[
\begin{split}
\Big\| d\Gamma_{l}\big(\widehat{u}_{t}e^{ip\cdot (x - i\e\nabla t)}\widehat{u}_{t}\big)\psi \Big\| &\leq C\| \cN \psi \| \\
\Big\| d\Gamma_{l}\Big(i\e\frac{d}{dt}\widehat{u}_{t}e^{ip\cdot (x - i\e\nabla t)}\widehat{u}_{t}\Big)\psi \Big\| &\leq C\| (\cN + \mathcal{K}) \psi\|\;.
\end{split}
\]
Therefore, we get
\[
\Big\|  i\e\frac{d}{dt} \cU_{N}^{(0)*}(t) \frac{1}{2N}\int dxdy\, V(x-y)a^{*}_{l}(u_{x})a^{*}_{l}(u_{y})a_{l}(u_{y})a_{l}(u_{x}) \cU_{N}^{(0)}(t) \psi \Big\| \leq C\| \cN (\cN + \mathcal{K}) \psi\|\;.
\]
The remaining contributions in $\text{III}$ can be estimated by following strategies similar to those we discussed so far; we omit the details. We find:
\be\label{eq:estGIII}
\Big\| i\e\frac{d}{dt} \cU_{N}^{(0)*}(t)\mathcal{C}_{N}\cU_{N}^{(0)}(t) \psi \Big\| \leq C\| (\cN + 1)(\cN + \mathcal{K})\psi \|\;.
\ee
Summing up (\ref{eq:estGI}), (\ref{eq:estGII}) and (\ref{eq:estGIII}), we finally obtain:
\[
\Big\| i\e\frac{d}{dt}\widehat{\mathcal{G}}_{N}(t) \psi \Big\| \leq C\| (\cN + 1)(\cN + \mathcal{K})\psi \|\;,
\]
which is finite for $\psi\in D(\cN(\mathcal{K} + \mathcal{N}))$.

\section{Propagation of semiclassical structure}\label{secpropagation}

This appendix is devoted to the proof of Proposition \ref{lem:prop}. The argument follows  \cite{BPS} with some  minor modifications. 

\begin{proof}[Proof of Proposition \ref{lem:prop}] Let $h_{HF}(t) = -\e^{2}\D + \r_{t}*V - X_{t}$ be the Hartree-Fock Hamiltonian, and let $\o_{N,t}$ denote the solution of the Hartree-Fock equation
\be
i\e\partial_{t}\o_{N,t} = [h_{HF}(t), \o_{N,t}]
\ee
with $\o_{N,0} = \o_{N}$. Since $u_{N,t} = \sqrt{1 - \o_{N,t}}$, $v_{N,t} = \sqrt{\o_{N,t}}$, we have
\be
i\e\partial_{t} v_{N,t} = [h_{HF}(t), v_{N,t}]\;,\qquad i\e\partial_{t} u_{N,t} = [h_{HF}(t), u_{N,t}]
\ee
with $u_{N,0} = u_{N} = \sqrt{1 - \o_{N}}$ and $v_{N,0} = v_{N} = \sqrt{\o_{N}}$. We first prove the propagation of the estimates for $v_{N,t}$. We compute
\[
i\e\frac{d}{dt}[x, v_{N,t}] = [x, [h_{HF}(t), v_{N,t}]] = [v_{N,t}, [h_{HF}(t), x]] + [h_{HF}(t), [x, v_{N,t}]].
\]
One eliminates the last term by conjugating $[x, v_{N,t}]$ with the two-parameter group $W(t;s)$ generated by the self-adjoint operator $h_{HF}(t)$ satisfying
\[
i\e\frac{d}{dt} W(t;s) = h_{HF}(t) W(t;s)\qquad \mbox{with $W(s;s) = 1$ for all $s\in\mathbb{R}$}.
\]
We have
\[
\begin{split}
i\e\frac{d}{dt} W^{*}(t;0) [x, v_{N,t}] W(t;0) & = W^{*}(t;0) [v_{N,t}, [h_{HF}(t), x]] W(t;0) \\
& = W^{*}(t;0)\Big( [v_{N,t}, -2\e^{2}\nabla] - [v_{N,t}, [X_{t},x]] \Big)W(t;0)
\end{split}
\]
where we used that $[-\e^{2}\D,x]= - 2\e^{2}\nabla$ and that $[\rho_{t}*V, x] = 0$. Therefore,
\[
\begin{split}
W^{*}(t;0) [x, v_{N,t}] W(t;0) & = [x, v_{N,0}] + \frac{1}{i\e}\int_{0}^{t}ds\, i\e\frac{d}{ds} W^{*}(s;0) [x, v_{N,s}] W(s;0) \\ 
& = [x, v_{N,0}] - \frac{1}{i\e}\int_{0}^{t}ds\, W^{*}(s;0) \Big( [v_{N,s}, 2\e^{2}\nabla] + [v_{N,s}, [X_{s},x]] \Big) W(s;0).
\end{split}
\]
This implies that
\be
\| [x, v_{N,t}] \|_{HS} \leq \| [x, v_{N,0}] \|_{HS} + \frac{1}{\e}\int_{0}^{t}ds\,\Big( \| [v_{N,s}, 2\e^{2}\nabla] \|_{HS} + \|[v_{N,s}, [X_{s},x]]\|_{HS} \Big).\label{p1}
\ee
To control the second term, notice that
\be
X_{s} = \frac{1}{N}\int dq\, \hat V(q)\,e^{ip\cdot x} \o_{N,s} e^{-ip\cdot x}.\label{p1b}
\ee
Using that $\| v_{N,s} \|_{op}\leq 1$, we find:
\be\label{p2}
\begin{split}
\|[v_{N,s}, [X_{s},x]]\|_{HS} & \leq \frac{1}{N}\int dq\, |\hat V(q)|\, \|[v_{N,s}, [e^{ip\cdot x} \o_{N,s} e^{-ip\cdot x}, x]]\|_{HS} \\
& \leq \frac{1}{N}\int dq\, |\hat V(q)|\, \|[v_{N,s}, e^{ip\cdot x}[\o_{N,s}, x]e^{-ip\cdot x}]\|_{HS} \\
& \leq \frac{2}{N}\int dq\, |\hat V(q)|\, \|[\o_{N,s}, x]\|_{HS} \\
& \leq \frac{C}{N}\| [v_{N,s}, x] \|_{HS}
\end{split}
\ee
where in the last step we used the identity 
\be\label{p1c}
[\o_{N,s}, x] = [ v_{N,s}v_{N,s}, x] = v_{N,s}[v_{N,s},x] + [v_{N,s},x]v_{N,s}
\ee
and  the triangle inequality. Substituting  the estimate (\ref{p2}) in (\ref{p1}), we derive
\be
\| [x, v_{N,t}] \|_{HS} \leq \| [x, v_{N,0}] \|_{HS} + C \int_{0}^{t}ds\,\Big( \| [v_{N,s},  \e\nabla] \|_{HS} +  N^{-2/3}\|[v_{N,s}, x]\|_{HS} \Big). \label{p1ba}
\ee
To control $[v_{N,s},\e\nabla]$, we start by writing
\[
\begin{split}
i\e\frac{d}{dt}[\e\nabla, v_{N,t}] & = [\e\nabla, [h_{HF}(t), v_{N,t}]] = [v_{N,t}, [h_{HF}(t), \e\nabla]] + [h_{HF}(t), [\e\nabla, v_{N,t}]] \\ 
& = [h_{HF}(t), [\e\nabla, v_{N,t}]] + [v_{N,t}, [\r_{t}*V, \e\nabla]] - [v_{N,t}, [X_{t}, \e\nabla]].
\end{split}
\]
As before, the first term can be eliminated after conjugating with the unitary operator $W(t;0)$. Hence, we find
\be
\| [\e\nabla, v_{N,t}] \|_{HS} \leq \| [\e\nabla, v_{N,0}] \|_{HS} + \frac{1}{\e}\int_{0}^{t}ds\, \Big( \| [v_{N,s}, [\r_{s}*V, \e\nabla]] \|_{HS} + \| [v_{N,s}, [X_{s}, \e\nabla]]\|_{HS} \Big).\label{p4}
\ee
The second term can be controlled by writing
\be\label{p1bb}
\begin{split}
\| [v_{N,s}, [V*\r_{s}, \e\nabla]] \|_{HS} & = \e \| [v_{N,s}, \nabla V * \r_{s}] \|_{HS}\\
& \leq \e\int dq\, |\hat V(q)|\, |q|\, |\hat \r_{s}(q)|\, \| [v_{N,s}, e^{iq\cdot x}] \|_{HS}\\
&\leq \e\int dq\, |\hat V(q)|\, |q|^{2}\, |\hat \r_{s}(q)|\, \| [v_{N,s}, x] \|_{HS} \\
& \leq C\e \| [v_{N,s}, x] \|_{HS}
\end{split}
\ee
where we used the bound $\| \hat \r_{s} \|_{\infty} \leq \|\r_{s}\|_{1} = 1$ and the assumption on the interaction potential. Let us now focus on the last term in (\ref{p4}). Rewriting $X_{s}$ as in (\ref{p1b}), we have
\be\label{p4b}
\begin{split}
\| [v_{N,s}, [X_{s}, \e\nabla]] \|_{HS} &\leq \frac{1}{N}\int dq\,|\hat V(q)|\, \| [v_{N,s}, [e^{iq\cdot x}\o_{N,s}e^{-iq\cdot x}, \e\nabla]] \|_{HS} \\
&\leq \frac{2}{N}\int dq\,|\hat V(q)|\,\| [e^{iq\cdot x}\o_{N,s}e^{-iq\cdot x}, \e\nabla ] \|_{HS} \\
&\leq \frac{2}{N}\int dq\,|\hat V(q)|\,\| [\o_{N,s}, \e\nabla ] \|_{HS} \\
&\leq \frac{4 \| \hat V \|_{1}}{N}\| [v_{N,s}, \e\nabla] \|_{HS}
\end{split}
\ee
where to get the third inequality we used that
\[
[e^{iq\cdot x}\o_{N,s}e^{-iq\cdot x},\e\nabla] = e^{iq\cdot x}[\o_{N,s}, \e(\nabla + iq)]e^{-iq\cdot x} = e^{iq\cdot x}[\o_{N,s}, \e\nabla]e^{-iq\cdot x}
\]
and to get the last inequality we used the analog of identity (\ref{p1c}) with $x$ replaced by $\e\nabla$. Substituting  the estimates (\ref{p1bb}), (\ref{p4b}) into (\ref{p4}), we get
\[
\| [\e\nabla, v_{N,t}] \|_{HS} \leq \| [\e\nabla, v_{N,0}] \|_{HS} + C\int_{0}^{t}ds\,\Big( \| [v_{N,s}, x] \|_{HS} + N^{-2/3}\| [v_{N,s}, \e\nabla] \|_{HS} \Big).
\]
Summing up this inequality with (\ref{p1ba}), using the assumption on the initial data and applying Gronwall's lemma, we obtain
\be
\| [x, v_{N,t}] \|_{HS} \leq K \exp(c|t|) N^{1/2}\e\;,\qquad \| [\e\nabla, v_{N,t}] \|_{HS} \leq K \exp(c|t|) N^{1/2}\e
\ee
for suitable positive constants $K,\,c$ only depending on $V$. This proves the propagation of the semiclassical estimates for $v_{N,t}$. To conclude the proof, we have to prove the propagation of the semiclassical estimates for $u_{N,t}$. The argument proceeds in the same way, with the following modifications. In the estimate for $\|[ u_{N,s}, x ]\|_{HS}$, in the step analogous to (\ref{p2}) one has to use that
\[
\| [\o_{N,s}, x] \|_{HS} = \| [1 - \o_{N,s}, x] \|_{HS} \leq 2\| [u_{N,t}, x] \|_{HS}
\]
where we used that
\be\label{p5}
[ 1 - \o_{N,s}, x] = [u_{N,s}u_{N,s}, x] = u_{N,s}[u_{N,s}, x] + [u_{N,s},x]u_{N,s}.
\ee
Similarly, in the estimate for $\| [u_{N,s}, \e\nabla] \|_{HS}$, in the step analogous to (\ref{p4b}) one has to use that  
\[
\| [\o_{N,s}, \e\nabla] \|_{HS} = \| [1 - \o_{N,s}, \e\nabla] \|_{HS} \leq 2\| [u_{N,s},\e\nabla] \|_{HS}
\]
where we used the analog of (\ref{p5}) with $x$ replaced by $\e\nabla$. Taking into account these small modifications, one derives
\be
\| [x, u_{N,t}] \|_{HS} \leq K \exp(c|t|) N^{1/2}\e\;,\qquad \| [\e\nabla, u_{N,t}] \|_{HS} \leq K \exp(c|t|) N^{1/2}\e
\ee 
for suitable constants $K,\,c$ only depending on $V$. This proves the propagation of the semiclassical structure for $u_{N,t}$, and concludes the proof of Proposition \ref{lem:prop}.
\end{proof}

\end{document}